\documentclass[reqno,english,11pt]{amsart}
\usepackage{bm}
\usepackage{amsmath,amsthm,verbatim,amssymb,amsfonts,amscd, graphicx,graphics}
\usepackage{t1enc}
\usepackage[mathscr]{eucal}
\usepackage{indentfirst}
\usepackage{booktabs}
\usepackage{pict2e}
\usepackage{epic}
\numberwithin{equation}{section}
\usepackage[footnotesize,bf]{caption}
\usepackage{fullpage}
\usepackage{epstopdf}
\usepackage{mathtools}
\usepackage[dvipsnames]{xcolor}
\usepackage{algorithmic}
\usepackage[ruled,vlined]{algorithm2e}
\usepackage{authblk}

\usepackage[titletoc]{appendix}
\usepackage{enumitem,multicol}

\usepackage{threeparttable}
\usepackage{subcaption} 
\usepackage{siunitx}
\usepackage{multirow}

\usepackage[foot]{amsaddr}

\newcommand{\ER}{Erd{\"o}s--R\'{e}nyi }

\usepackage[colorlinks,linkcolor = BrickRed, citecolor=blue]{hyperref} 
\usepackage{etoolbox,cleveref}
\AtBeginEnvironment{appendices}{\crefalias{section}{appendix}}

\newlength{\leftstackrelawd}
\newlength{\leftstackrelbwd}
\def\leftstackrel#1#2{\settowidth{\leftstackrelawd}%
{${{}^{#1}}$}\settowidth{\leftstackrelbwd}{$#2$}%
\addtolength{\leftstackrelawd}{-\leftstackrelbwd}%
\leavevmode\ifthenelse{\lengthtest{\leftstackrelawd>0pt}}%
{\kern-.5\leftstackrelawd}{}\mathrel{\mathop{#2}\limits^{#1}}}

 \theoremstyle{plain}
\newtheorem{Th}{Theorem}[section]
\newtheorem{Lemma}[Th]{Lemma}

\theoremstyle{definition}

\newtheorem{Rem}[Th]{Remark}
\newtheorem{?}[Th]{Problem}

\def\E{{\mathbb E}}

\def\H{{\mathsf H}}

\def\P{{\mathbb P}}

\def\A{{\mathcal A}}

\def\G{{\mathsf{G}}}

\def\e{{\varepsilon}}

\def\PC{{\mathcal{PC}}}
\def\V{{\mathcal{V}}}

\DeclareMathOperator*{\argmin}{arg\,min}

\begin{document}

\title{Finding planted cliques using gradient descent}

\author[Reza Gheissari]{Reza Gheissari$^*$}
\thanks{$*$; Department of Mathematics, Northwestern University. gheissari@northwestern.edu}

\author[Aukosh Jagannath]{Aukosh Jagannath$^\dagger$}
\thanks{$\dagger$: Department of Statistics and Actuarial Science, Department of Applied Mathematics, and Cheriton School of Computer Science, University of Waterloo. a.jagannath@uwaterloo.ca}

\author[Yiming Xu]{Yiming Xu$^\ddagger$}
\thanks{$\ddagger$: Department of Mathematics, University of Kentucky. yiming.xu@uky.edu}

\maketitle

\thispagestyle{empty}

\begin{abstract}
        The planted clique problem is a paradigmatic model of statistical-to-computational gaps: the planted clique is information-theoretically detectable if its size $k\ge 2\log_2 n$ but polynomial-time algorithms only exist for the recovery task when $k= \Omega(\sqrt{n})$. By now, there are many algorithms that succeed as soon as $k = \Omega(\sqrt{n})$. Glaringly, however, no black-box optimization method, e.g., gradient descent or the Metropolis process, has been shown to work. In fact, Chen, Mossel, and Zadik recently showed that any Metropolis process whose state space is the set of cliques fails to find any sub-linear sized planted clique in polynomial time if initialized naturally from the empty set.
        We show that using the method of Lagrange multipliers, namely optimizing the Hamiltonian given by the sum of the objective function and the clique constraint over the space of all subgraphs, succeeds. In particular, we prove that Markov chains which minimize this Hamiltonian (gradient descent and a low-temperature relaxation of it) succeed at recovering planted cliques of size $k =  \Omega(\sqrt{n})$ if initialized from the full graph. Importantly, initialized from the empty set, the relaxation still does not help the gradient descent find sub-linear planted cliques. We also demonstrate robustness of these Markov chain approaches under a natural contamination model.     
\end{abstract}

\section{Introduction}
The most commonly used method for solving constrained optimization problems is the method of Lagrange multipliers. Here one runs gradient descent on an energy function, or Hamiltonian, which is given by the sum of the objective function and the constraint. 
In many of the central problems arising out of the recent literature in computational complexity of statistical inference, however, there is a substantial gap between the performance guarantees for this black-box method and problem-specific methods. 
From this perspective, it is natural to try to close the gap in arguably the central problem in the field, namely the planted clique problem.

The planted clique problem is the algorithmic task of finding a planted clique with $k$ vertices inside an \ER $\mathsf{G}(n,\frac{1}{2})$ random graph. This task was first introduced in the work of Jerrum \cite{jerrum1992large} as a planted version of the maximum clique problem proposed by Karp \cite{karp1976probabilistic}. Since then, the planted clique problem has become a central problem in average-case complexity in its own right as it is one of the simplest models exhibiting a \emph{statistical-to-computational gap}: the clique is information-theoretically recoverable if $k\ge 2\log_2 n$, but it is expected to be algorithmically intractable to recover in polynomial time whenever $k=o(\sqrt{n}).$ Indeed, there is a substantial and influential line of recent work in which various problems are shown to have average case reductions to the planted clique problem (more precisely, its detection analogue) including the sparse PCA problem~\cite{BerthetRigollet} and certain community detection problems~\cite{Hajek-et-al}, and more broadly a hierarchy of reductions of~\cite{BrennanBresler} to planted clique with generalized priors.

By now, many algorithms are known to succeed for recovery of the planted clique when $k  = \Omega( \sqrt{n})$. 
Kucera~\cite{kuvcera1995expected} noticed that if $k\ge C \sqrt{n\log n}$, the $k$ vertices of the largest degree typically form the planted clique. The $\sqrt{n}$ threshold was then attained by a spectral algorithm in~\cite{alon1998finding}; see also the refinement of~\cite{McSherry} and semi-definite programs of \cite{feige2000finding,ames2011nuclear}. Other fast algorithms have been provided for the planted clique problem including multi-stage algorithms of~\cite{feige2010finding, dekel2014finding} and a message-passing-based algorithm by~\cite{DeshpandeMontanari}. On the hardness side, there have been low-degree semidefinite programming refutations~\cite{FeigeKrauthgamer}, statistical query lower bounds~\cite{Feldman-Statistical-query-PC}, and sum-of-squares lower bounds~\cite{MekaPotechinWigderson-SOS-lowerbounds,DeshpandeMontanari,HKP-degree-4-SOS} culminating in~\cite{Barak-SoS-refutation}, together indicating canonical families of ``low-degree'' algorithms cannot succeed at recovering the planted clique in polynomial time when $k=o(\sqrt{n})$.

Prior to these hardness results, the main argument for algorithmic intractability when $k=o(\sqrt{n})$ was the work of Jerrum~\cite{jerrum1992large}. Jerrum observed that gradient descent restricted to the constraint set (i.e., only moving on the space of cliques) fails to find the planted clique; moreover, he showed that a  Markov Chain Monte Carlo (MCMC) relaxation, namely the Metropolis process on the set of cliques is slow to mix. 
Since problem-specific algorithms work when $k=\Omega(\sqrt{n})$, Jerrum posited that if the Metropolis process could not also find large cliques above this scale it would be a ``severe indictment of the Metropolis process''. Despite the many algorithms that do succeed when $k = \Omega( \sqrt{n})$, for thirty years, there was no matching positive result for these off-the-shelf approaches.  

This question was revisited in recent years, with the work of Gamarnik and Zadik~\cite{GamarnikZadik} finding an annealed \emph{overlap gap property} (see e.g., the survey~\cite{GamarnikOGP}) in clique space that persisted well beyond the $\sqrt{n}$ threshold, suggesting that perhaps MCMC algorithms that move purely on cliques would not succeed even in the easy regime for the problem (at least from worst-case initializations). 
Remarkably, in subsequent work, Chen, Mossel, and Zadik~\cite{chen2023almost} proved that gradient descent and Metropolis processes restricted to the set of cliques fail to find the planted clique even when initialized from the empty configuration---the most natural uninformed initialization in the space of cliques---and that this failure holds whenever $k =O(n^{\alpha})$ for any $\alpha<1$!
They went further to obtain similar results even for simulated tempering.
This raised the question of whether such approaches were simply ill-suited to one of the central problems of average-case complexity and statistical-to-computational gaps. On the other hand, in~\cite{GamarnikZadik}, it was asked whether a relaxation of the state space could perhaps salvage the performance of such methods and attain the $\sqrt{n}$ threshold.

Our main result answers this in the affirmative.
We show that the black-box method of Lagrange multipliers succeeds from a natural uninformed start. 
 Specifically, we allow the Markov chain to evolve on the space of \emph{all subsets} of $V$ by adding a penalty term to the Hamiltonian for the missing edges; see~\eqref{myH}. We show that when initialized from the full vertex set, $V$, both the gradient descent and its low-temperature MCMC analogue recover the planted clique in linearly many steps, as long as $k = \Omega(\sqrt{n})$: see Theorem~\ref{mainthm:success}. Note, that as we have relaxed the problem to the space of all subgraphs, there are now \emph{two} natural uninformed initializations: the empty set and the full graph. Importantly, we find that the choice of which uninformed start is crucial to the success of the algorithm: like the Metropolis process on cliques, gradient descent is unable to find the planted clique when started from the empty set; see Theorem~\ref{mainthm:failure}.

\subsection{Main results}
The planted clique distribution $\mathsf{G}(n,\frac{1}{2},k)$ is the distribution over random graphs $G =(V,E)$ in which $V= [n]$, and the edge set $E$ is randomly drawn as follows: pick a subset $\mathcal P \mathcal C \subset [n]$ of $k$ vertices uniformly at random, and include every internal edge of $\PC$ with probability $1$, while independently including every other edge with probability $\frac{1}{2}$.

The planted clique problem is the algorithmic task of recovering the vertex set $\PC$, given $G$.  Above the information theoretic threshold, this is well-known to be equivalent to the maximum clique problem. That is, one seeks to find 
\begin{align*}
     \max_{U\subset V}\,  &|E(U)| \\
     \text{subject to: }\,   &|E(U)| = \binom{|U|}{2} 
\end{align*}
where $E(U)$ is the edge-set of the subgraph induced by $G$ on $U$. 
The method of Lagrange multipliers for this problem amounts to optimizing the following Hamiltonian over $U \subset V=[n]$: for $\gamma>1$, let 
\begin{align}\label{myH}
    H(U) = H_{G,\gamma}(U)= - |E(U)| + \gamma \Big[ \binom{|U|}{2} - |E(U)|\Big]\,.
\end{align}
The restriction to $\gamma>1$ is because this ensures that $H$ is minimized by $\PC$: see Theorem~\ref{thm:global optimum}.
When restricted to cliques this is essentially the energy function of~\cite{jerrum1992large}.\footnote{Jerrum's Hamiltonian had the number of vertices in $U$ as opposed to the number of edges $E(U)$. For cliques, these are effectively the same from an optimization standpoint.} The only difference is that the state space is expanded to all subsets, and the constraint is imposed by the Lagrange multiplier $\gamma$ penalizing the number of missing internal edges of $U$. 
A related relaxation---with a fixed number of non-clique vertices---was suggested in~\cite{GamarnikZadik} as a possible approach to circumventing some bottlenecks of the energy landscape.

We consider the gradient descent as well as its low-temperature, local Markov chain relaxation. In what follows, we write $U \sim U'$ if they are at Hamming distance at most $1$ from one another. Gradient descent on~\eqref{myH} is the following Markov chain $\{S_i\}_{i\geq 0}$: initialize from some $S_0\subset V$;
\begin{enumerate}[label=(\arabic*)]
    \item For every $i\ge 1$, if $H(S_{i-1})> \min\{H(U):U\sim S_{i-1}\}$, draw $S_{i}$ uniformly at random from $\argmin\{H(U): U\sim S_{i-1}\}$ (this may be non-singleton in the case of ties). 
    \item Else, let $S_{i}= S_{i-1}$ (terminating the process). 
\end{enumerate}
Evidently, this process will eventually terminate in an absorbing state (local minimum) of $H$. Let us note here that whenever $k \gg \log n$, the unique global minimizer of~\eqref{myH} is the planted clique itself with probability tending to one, but also that there exist many small, non-$\mathcal P\mathcal C$, local minimizers of the Hamiltonian~\eqref{myH}. In particular,  the landscape is not at all convex. See Section~\ref{sec:energy} where we study the energy landscape in more detail.

We also consider a positive temperature, Gibbs sampler, relaxation of the gradient descent. It is natural to look at Markov chains whose invariant measure is proportional to $\exp(-\beta H(U))$ for the same Hamiltonian. Note, however, that the zero temperature limits of popular chains with this invariant measure, such as Metropolis or Glauber with respect to $H$, are \emph{not} given by gradient descent; their zero-temperature limits make uniform-at-random choices over \emph{all} lower-energy neighbors. To correct for this entropic effect, we instead consider the random walk on the hypercube with transition probabilities given by the local Gibbs probabilities. Namely, we consider the discrete-time Markov chain $\{S_i^\beta\}_{i\ge 0}$ with transition probability from $W$ to $U$ given by
\begin{align}
    P(W,U) =  \begin{cases}\frac{e^{ - \beta H(U)}}{Z(W)} & U\sim W \\ 0 & \text{else}
       \end{cases}\, ,\label{glauber}
    \end{align}
where $Z(W)=\sum_{U'\sim W}e^{ - \beta H(U')}$.
The zero-temperature ($\beta \to \infty$) limit is indeed the gradient descent chain for $H(U)$. This Markov chain is itself a Gibbs sampler as it is easily checked to be reversible with respect to the tilted Gibbs measure $\nu(U)\propto  e^{ - \beta H(U) + \log Z(U)}$.  When $\beta$ is large, $\nu(U)$ is concentrated on $\PC$, 
and thus this is also a reasonable approach to finding the planted clique.  Intuitively, one can view this new measure $\nu$ as the Gibbs measure where the Hamiltonian is modified by the free energy of the 1-neighborhood of $U$. It could be interesting to compare the performance of this Gibbs sampler to the standard Glauber dynamics for $H$ in this and other discrete optimization problems.  

Our main result is that initialized from $S_0 = S_0^\beta = V$, the gradient descent, and the low-temperature (large $\beta$) Gibbs sampler, described above, both find the planted clique in $O(n)$ steps. 

\begin{Th}\label{mainthm:success}
Suppose $\gamma>3$. 
For every $\e>0$, there exists $C(\e,\gamma)>0$ such that for all $k\geq C\sqrt{n}$, with probability at least $1-\e$, the gradient descent $S_t$ initialized from $S_0 = [n]$ achieves 
\begin{align*}
S_t = \PC \qquad \text{ for all \,$t\ge n+2k$}\,.
\end{align*}
The same holds for the low-temperature chain $S_t^\beta$ for all $n+2k \le t \le n^{k/C}$ if $\beta \ge C \log n$.
\end{Th}

\begin{Rem}
    We are not careful here about the constant $C$ and its dependence on $\e$. Note that one could use 
    a boosting scheme of~\cite{alon1998finding} which takes any algorithm that works at $k \ge C\sqrt{n}$ to one that work when $k\ge \frac{C}{\sqrt 2}\sqrt{n}$ by paying a $\sqrt{n}$ in the running time.
\end{Rem}

The mechanism behind the success proved by Theorem~\ref{mainthm:success} when $S_0 = V$, is roughly described as follows: it initially peels off vertices of the lowest degree within $S_t$, then after some time, vertex additions may be possible. At that point, it typically oscillates in size, adding $\PC$ vertices while removing extraneous non-$\PC$ vertices, before eventually converging to the full $\PC$. 
Analysis of the landscape of $H$~\eqref{hamiltonian} reveals that the beginning and end of this trajectory resemble the two-stage algorithm of Feige and Ron~\cite{feige2010finding}. We emphasize, however, that in the middle the trajectories will differ somewhat, and that this is produced by a black box constrained optimization approach.

We remark that while Theorem~\ref{mainthm:success} only guarantees the success of the positive temperature chain for $\beta = \Omega( \log n)$, we expect that at $\beta$ sufficiently large, but $O(1)$, the same success should hold. The proof uses stochastic domination for vertex degrees, which seems too fragile to handle the situation in which $\beta = O(1)$, and non-energy-minimizing moves are taken with uniformly positive probability.  

While tools like spectral gaps, mixing times, and overlap gap properties are very useful in proving refutation results for Markov chains with worst-case initializations, Theorem~\ref{mainthm:success} demonstrates the power of a well-chosen (though still completely uninformative) initialization to help a Markov chain succeed at sampling and optimization when worst-case mixing times are slow (e.g., due to the presence of local minima in $H$ with zero overlap with $\PC$, per Theorem~\ref{prop:many-local-minimum}). 
The success of gradient descent initialized uninformatively, even when worst-case initialization fails, was for instance also leveraged in MCMC analysis of the tensor PCA problem~\cite{MontanariRichard} in~\cite{BGJ20}.

\subsubsection{The importance of the initialization}
It is natural to wonder if these Markov chains are successful from every $G$-independent initialization, or if there are bottlenecks in the space that the full initialization is circumventing. It turns out that when initialized from $S_0 = \emptyset$, they face the same obstruction that Metropolis processes on cliques do, even well beyond the $\sqrt{n}$ threshold. 

\begin{Th}\label{mainthm:failure}
If $\gamma>1$ and $k \le n^\alpha$ for some $\alpha< 1$, with probability $1-o(1)$, the gradient descent $S_t$ initialized from $S_0 = \emptyset$ fails to find $\PC$ in the planted clique model $\G(n, \frac{1}{2}, k)$, i.e., it absorbs in $\text{polylog}(n)$ steps into a configuration $S_\infty$ of size $O(\log n)$ with no intersection with $\PC$.
\end{Th}

\begin{figure}[t]
 \centering 
\begin{subfigure}{0.4\textwidth}{\includegraphics[width=\linewidth, trim={0cm 0.0cm 0cm 0cm},clip]{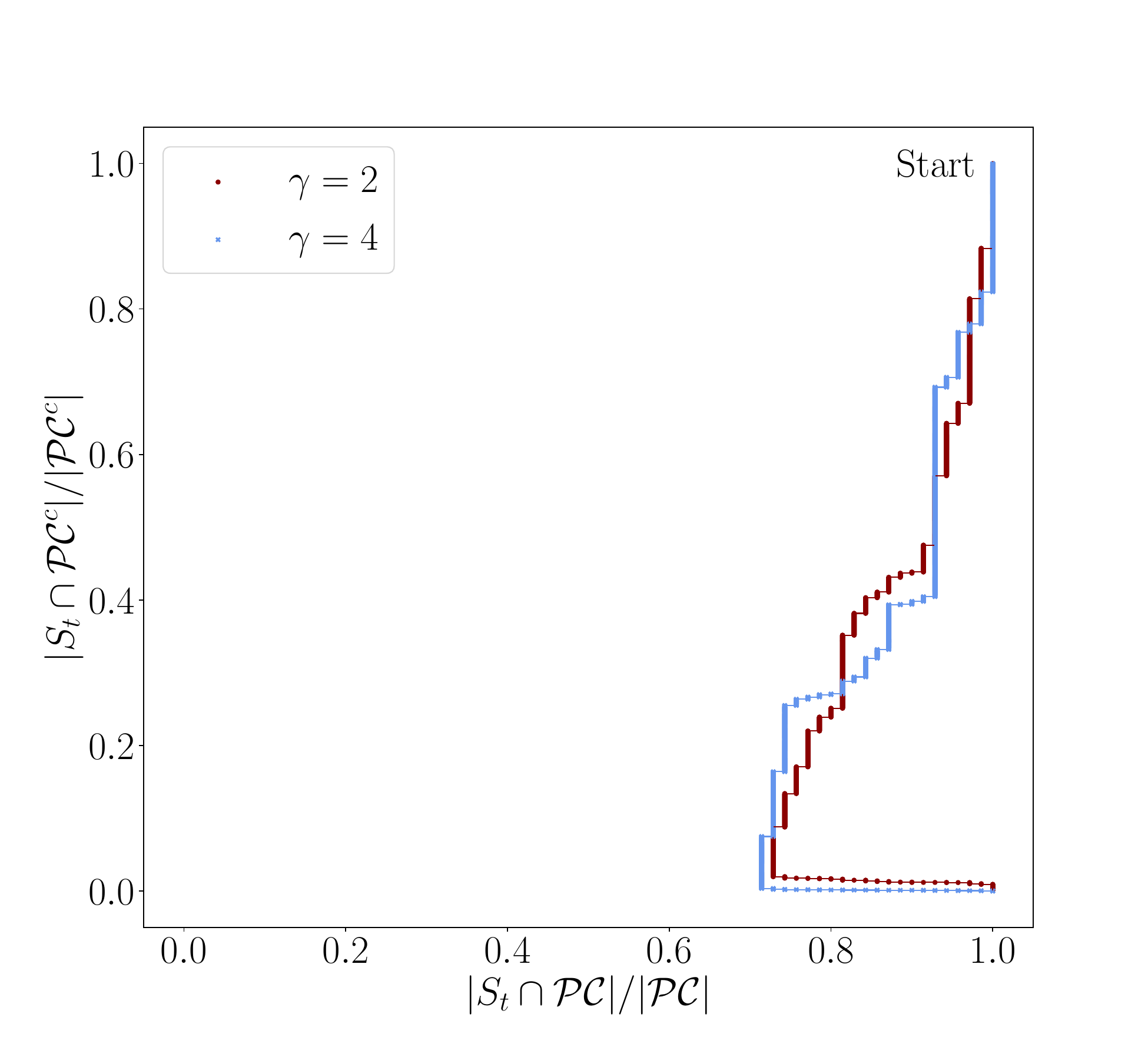}}\label{11:(d)}
\end{subfigure} \hspace{1 cm}
\begin{subfigure}{0.4\textwidth}{\includegraphics[width=\linewidth, trim={0cm 0.0cm 0cm 0cm},clip]{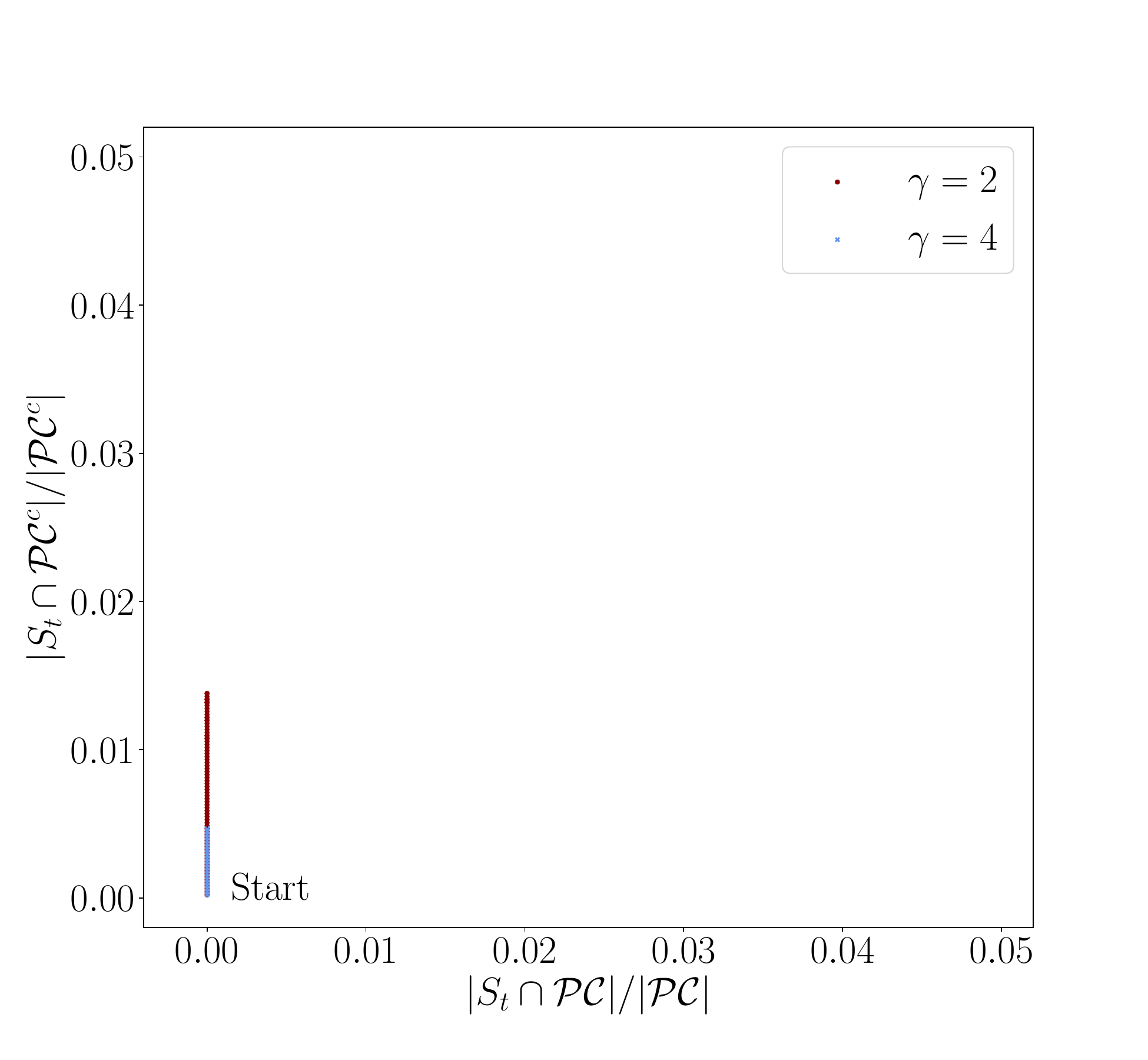}}\label{11:(s)}
\end{subfigure} 
\caption{Simulated trajectories of the relative ratio of $\PC$ and non-$\PC$ vertices in $S_t$ while applying the gradient descent with full-graph initialization (left) and empty-set initialization (right) to find the planted clique in $\G(5000, \frac{1}{2}, 70)$ under different values of $\gamma$. The $x$-axis is the overlap with $\PC$ and the $y$-axis is the overlap with $\PC^c$.
Both trajectories in the left plot start from $(1, 1)$ and terminate in $\PC$, whereas both trajectories in the right plot start from $(0, 0)$ and stop at local minima that do not have any overlap with $\PC$. 
}\label{fig:0}
\end{figure}

Theorem~\ref{mainthm:failure} demonstrates that not only do these chains fail from worst-case initialization, they even fail from the natural uninformative initialization of $S_0 = \emptyset$. Importantly, however, with the state space relaxed to all subsets, $S_0 = V$ is an alternative natural uninformative initialization that does not face this bottleneck. The source of this discrepancy can roughly be described as follows: on the way up from $\emptyset$, the chain doesn't behave that differently from one that moves purely on cliques, and it therefore gets trapped by the large entropy of $\Theta(\log n)$ (near-)cliques that are disjoint from $\PC$.  Conversely, from $S_0  = V$ it can descend straight to $\PC$ without facing any entropy-induced bottleneck: i.e., the algorithm is starting on the ``right side'' of the overlap gap. This dichotomy of behaviors depending on the initialization is shown in Figure~\ref{fig:0}. 

To make this into a proof, we show that the gradient descent initialized from $S_0 = \emptyset$ performed on $G$ is coupled with a gradient descent process on \ER $G_0$ (before the edges of $\PC$ were forced to be included). This latter process absorbs quickly into one of the many near-clique local minima of $H$, with zero overlap with $\PC$. This approach for proving failure by coupling the dynamics with one in an unplanted (zero-signal) model may be of independent interest.

\subsubsection{Robustness to adversary}
Let us end by briefly discussing the robustness of the gradient descent and positive temperature chain to adversarial planting. Robustness to an adversary tweaking non-$\PC$ edges has been studied extensively for the planted-clique model since it has been suggested that the fact that the highest degree vertices recover $\PC$ nearly down to the predicted algorithmic bound makes the planted clique problem less realistic than its robust versions. Some of the different robust variants that have been introduced include the monotone adversaries model of~\cite{BlumSpencer} and the semi-random model of~\cite{Feige-Semirandom}; information-theoretic and algorithmic thresholds for semi-random graph problems have seen much attention, e.g.,~\cite{CojaOghlan-semirandom,steinhardt2018does,Charikar-semirandom,McKenzie-et-al-semi-random,Buhai-Kothari-Steurer}. 

We consider a weaker form of robustness, but one strong enough that the adversary can still change the set of high-degree vertices to differ significantly from $\PC$, for instance. Namely, we allow an adversary to change the edge-probability $1/2$ to $q>1/2$ for some $m=O(n^{3/4-o(1)})$  number of the vertices in $G$. Note that with this type of modification, the largest clique in $G$ will remain $\PC$ when $\PC\gg \log n$, but for instance, the highest degree vertices will contain those modified by the adversary if $q>1/2$, rather than being $\PC$, even well into the algorithmically tractable regime.

In Section~\ref{sec:rob}, we prove that the gradient descent and positive temperature chain for the Hamiltonian~\eqref{myH} are robust to such adversaries: see Theorem~\ref{lemma:n1+v}. Namely, they will still recover $\PC$ in linearly many steps. It would be of interest to explore MCMC approaches to the semi-random graph problems in the abovementioned literature, as well as on the hosts of other combinatorial optimization tasks, where problem-specific algorithms are well-understood but MCMC from natural uninformative initializations are not.

\section{The energy landscape}\label{sec:energy}
In this section, we collect results about the energy landscape of~\eqref{myH}. Our first result is that~$\PC$ is the global minimizer of~\eqref{myH}; this will follow from a concentration estimate on the degree counts between subsets $U$ and their intersections with $\PC^c$---Lemma~\ref{lemma:000}.

\begin{Th}[Energy landscape: global minimum]\label{thm:global optimum}
For any $\gamma>1$ and $0<\alpha\leq 1$, if $k \geq n^\alpha$, then with probability $1-o(1)$, $\argmin_{U\subseteq [n]} H_{G,\gamma}(U) = \PC$.
\end{Th} 

The other important landscape result for~\eqref{myH} we can establish is that the landscape exhibits \emph{complexity} in the sense that it has $n^{O(\log n)}$ local minima that are subsets of size $O(\log n)$, and that are completely uncorrelated with the planted clique $\PC$. We say a subset $U$ is a (strict) local minimum of $H$ if for all $W\sim U: W\ne U$, we have $H(W)>H(U)$. 
Let $h(p) = -p\log_2 p -(1-p)\log_2(1-p)$ be the binary entropy function and $\kappa = \frac{\gamma}{1+\gamma}\in (\frac{1}{2}, 1)$.

\begin{Th}[Energy landscape: local minima]\label{prop:many-local-minimum}
If $\gamma>9$ (whence $h(\kappa)<\frac{1}{2}$), then for any $c\in(\frac{1}{1-h(\kappa)}, 2)$, with probability at least $1-o(1)$, there are at least $n^{(1-\frac{1}{2}c(1-h(\kappa)) +o(1))m}$ many local minimizers of $H_{G,\gamma}$ with size $m=c\log_2 n$ with empty intersection with $\PC$. 
\end{Th}

Landscape complexity has been extensively studied and used as a heuristic explanation for hardness of optimization in the spin glass literature (see e.g., the works of~\cite{ABA13,ABC13}), and more recently for tensor PCA~\cite{BMMN17}, and even risk landscapes of generalized linear models~\cite{pmlr-Maillard}. Intuitively, the complexity of the landscape~\eqref{myH} at subsets of size $O(\log n)$ is consistent with our success and failure results. There is complexity in between the empty-set and the $\PC$, and the gradient descent fails from that initialization; on the other hand our proof of Theorem~\ref{thm:global optimum} demonstrates that there are no local minima of sizes larger than $O(\log n)$ except the planted clique itself, so there is no such complexity blocking descent from the full initialization: see Remark~\ref{rem:no-complexity-above}. See Figure~\ref{fig:1} for a visualization of the complexity landscape.

 \begin{figure}[t]
  \centering 
 \includegraphics[width=0.42\linewidth, trim={0cm 0cm 0cm 0cm},clip]{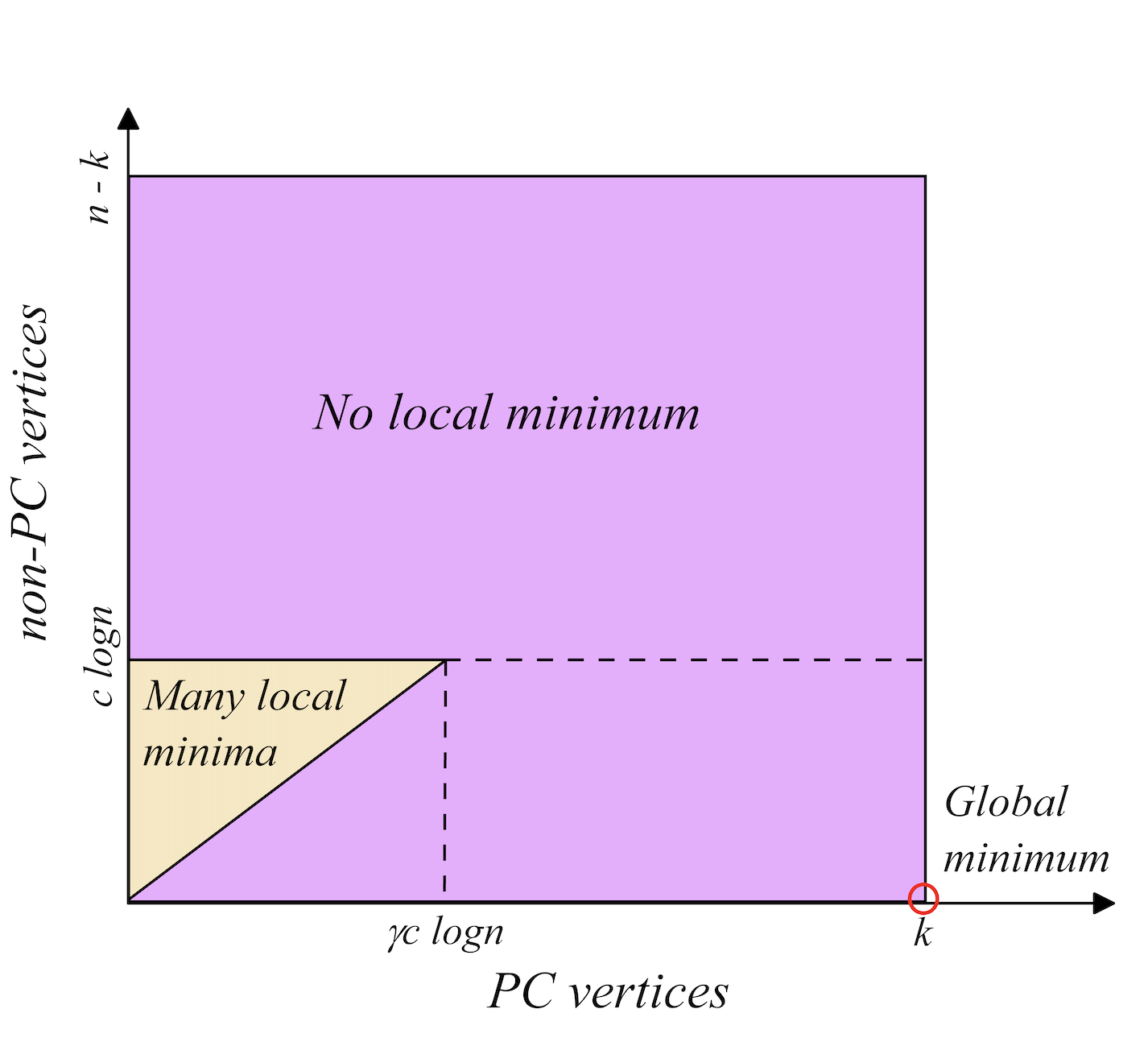}
 \caption{Phase diagram in terms of $U\cap \PC$ and $U\cap \PC^c$, depicting the regions where $H(U)$ has complexity and admits local minima. The global minimum (circled red) is exactly $\PC$.} \label{fig:1}
 \end{figure}

Theorem~\ref{thm:global optimum} is proved in Subsection~\ref{subsec:PC-is-global-min}.
Since Theorem \ref{prop:many-local-minimum} will not be used in the algorithmic results, its proof is postponed to Section~\ref{sec:proofs}. Its proof goes via a second-moment method.  

\subsection{Planted clique is the global energy minimizer}\label{subsec:PC-is-global-min}
Without loss of generality and for convenience, in the remainder of the paper, let $\PC = [k]:=\{1, \ldots, k\}$ with $k \geq n^\alpha$ for some $\alpha\in (0, 1)$, and let $E$ denote the edge set of the graph $G\sim \mathsf{G}(n,\frac{1}{2},k)$. 
For any $U\subseteq [n]$, we write $U = U_1\cup U_2$, where $U_1 = U\cap [k]$, $U_2 = U\cap [k]^c$, and define $H(U_1, U_2) = H(U)$. 
When $U$ is fixed, we typically use $n_1 = |U_1|$ and $n_2 = |U_2|$. 
We will frequently use the following energy-change relations when adding/removing a vertex to/from $U$:
\begin{align}
H(U\cup \{x\}) - H(U) &= -(1+\gamma)|E(x, U)| + \gamma |U|& x\in U^c\label{pl1}\\
H(U\setminus \{z\}) - H(U) &= (1+\gamma)|E(z, U)| - \gamma (|U|-1)&z\in U\label{pl2}.
\end{align}
where throughout the paper, we use $E(A,B)$ to denote the set of edges between $A$ and $B$.  
In what follows, we always assume $n$ is taken to be sufficiently large.  

The proof of Theorem \ref{thm:global optimum} is based on the following lemma, which will be frequently appealed to. This lemma is based on standard concentration bounds and union bounds and its proof is deferred to Section~\ref{sec:proofs}. 

\begin{Lemma}\label{lemma:000}
There exists an absolute constant $c>0$ such that for every $k\ge 0$, with probability $1-o(n^{-2})$, for all pairs $(U', U)$ satisfying $U'\subseteq U\subseteq [n]$ with $|U'|\geq c\log n$, 
\begin{align}\label{eq:degree-concentration}
\left(1-\delta\right)\E[\deg(U', U)]\leq \deg(U', U)\leq \left(1+\delta\right)\E[\deg(U', U)],
\end{align}
where $\deg(U', U) = \sum_{x\in U'} | E(x,U)|$ is the total degree of $U'$ in $U$ and $\delta\geq \sqrt{\frac{96\log n}{|U'|}}$.
Consequently,  
\begin{align}
\min_{x\in U'}|E(x, U)|\leq\left(1+\delta\right)\frac{\E[\deg(U', U)]}{|U'|}.\label{panger}
\end{align}
\end{Lemma}

Lemma \ref{lemma:000} holds for all $(U', U)$ in a wide regime simultaneously with probability tending to one.
Sometimes it is also useful to have the following sharper bound in a narrower regime. 
\begin{Lemma}\label{lemma:000+}
There exists an absolute constant $\rho>0$ such that, for any $\alpha\in (0, 1)$ and $n^\alpha \le k \le \frac{n}{\log n}$, with probability at least $1-o(1)$, for all $U$ with $|U\cap [k]^c|\geq\frac{n}{\log n}$ (taking $U' = U\cap [k]^c$ in Lemma \ref{lemma:000}), 
\begin{align*} 
\min_{x\in U\cap [k]^c}|E(x, U)|\leq \frac{\E[\deg(U\cap [k]^c, U)]}{|U\cap [k]^c|} + \rho\sqrt{n} = \frac{1}{2}(|U|-1) + \rho\sqrt{n}.
\end{align*}
\end{Lemma}

With these inputs, we can show that indeed $\PC$ is the global minimizer of $H$.  

\begin{proof}[\textbf{\emph{Proof of Theorem \ref{thm:global optimum}}}]
We first claim that there exists $c(\gamma)>0$ such that with probability at least $1-2n^{-2}$, for every $U_1$, 
\begin{align}
\min_{|U_2|> c\log n}H(U_1, U_2)> \min_{|U_2|\leq c\log n}H(U_1, U_2). \label{large-large-regime}
\end{align}
The proof of~\eqref{large-large-regime} follows from Lemma \ref{lemma:000} and is deferred momentarily. 
Meanwhile, for fixed $U_2$ and $x\in [k]\setminus U_1$,  by~\eqref{pl1}, if $n_1>\gamma n_2$,
\begin{align}
H(U_1\cup\{x\}, U_2) - H(U_1, U_2) \leq - n_1 + \gamma n_2<0. \label{obs}
\end{align}
Combining \eqref{large-large-regime} and \eqref{obs}, we conclude that, with probability at least $1-n^{-2}$, 
\begin{align}
\min_{U_1, U_2}H(U_1, U_2) &= \min\left\{\min_{|U_2|\leq c\log n}H([k], U_2), \min_{|U_1|\leq\gamma |U_2|\leq \gamma c\log n}H(U_1, U_2)\right\}.
\end{align}
Now notice that $H([k],\emptyset)$ is an upper bound for the first minimum, and equals $- \binom{k}{2}$. 
At the same time, 
\begin{align*}
\min_{|U_1|\leq\gamma |U_2|\leq \gamma c\log n} H(U_1, U_2)\geq - {(1+\gamma)c\log n\choose 2}>- \binom{k}{2}\,
\end{align*}
so as long as $k>(1+\gamma)c\log n$.
Consequently, 
\begin{align}
    \min_{U} H(U) = \min_{U_1, U_2}H(U_1, U_2) = \min_{|U_2|\le c\log n} H([k],U_2)\,. \label{hamiltonian}
\end{align}

To show that this minimum is attained with $U_2 = \emptyset$, start by letting $\tau = \frac{1}{2}(\frac{1}{2}+\frac{\gamma}{1+\gamma})>\frac{1}{2}$ so that $(1+\frac{1}{\gamma})\tau<1$. 
Applying Hoeffding's inequality for each non-$\PC$ vertex followed by a union bound, we have $\max_{x\in [k]^c}|E(x, [k])|\leq \tau k$ holds with probability at least $1-n^{-2}$. 
On this event, for any $U = [k]\cup U_2$ with $|U_2|\le c\log n$, and any $x\in U_2\neq\emptyset$, by~\eqref{pl2}
\begin{align}
H([k], U_2\setminus\{x\})-H([k], U_2) & = (1+\gamma)|E(x, U)| - \gamma(|U|-1)\label{obs1}\\
=\ & (1+\gamma)|E(x, [k])| -\gamma k+ (1+\gamma)|E(x, U_2)|-\gamma(|U_2|-1)\nonumber\\
\leq\ & \gamma\big[\big(1+1/\gamma\big)\tau-1\big]k + c\log n<0,\nonumber
\end{align}
where the last step is followed by $k\gg\log n$. 
This combined with \eqref{hamiltonian} yields the desired result. 

It remains to prove \eqref{large-large-regime}.
If we let $\delta = \sqrt{96/c}$ 
in Lemma \ref{lemma:000} (such a choice of $\delta$ works for all $|U'|\geq c\log n$ in Lemma \ref{lemma:000}), 
then for any $U_1, U_2$ with $|U_2|>c\log n$, with probability at least $1-2n^{-2}$, there exists $x\in U_2$ such that $|E(x, U)|\leq \frac{1+\delta}{2}(n_1+n_2-1)$. Consequently, 
\begin{align*}
&H(U_1, U_2\setminus\{x\})-H(U_1, U_2)\leq -\gamma\big[1-\tfrac{1}{2}\big(1+\tfrac{1}{\gamma}\big)(1+\delta)\big](n_1+n_2-1)\,.
\end{align*}
This is strictly negative if $\delta$ is sufficiently small, i.e., $c$ is sufficiently large. 
\end{proof}

\begin{Rem}\label{rem:no-complexity-above}
The above proof actually shows that in both $\left\{U: |U_2|>c\log n\right\}$ and $\{U: |U_1|>\gamma |U_2|\}$ portions of the state space, the Hamiltonian $H$ has no local minimum (in fact no absorbing state)  besides the global minimum. 
On the other hand, as we will show in Section \ref{sec:empty-init}, $H$ has at least one local minimum in the regime $\{|U_2|\leq c\log n\}$ with high probability, and in fact when $\gamma$ is at least a large constant, there are many local minimizers in the portion of the state space $\{U: |U_2|\leq c\log n\}$.  
\end{Rem}

\section{Recovery from the full-graph initialization}\label{sec:full-graph}

In this section, we analyze the dynamics of the gradient descent and positive temperature chain with full-graph initialization $S_0 = V$ and prove Theorem~\ref{mainthm:success}.
Though the full-graph initialization is uninformative, it provides sufficient time for the algorithm to explore the global structure of the graph in the portion of the state space where there are no local minima to trap the gradient descent. 
The crux of the proof of Theorem~\ref{mainthm:success} are the following two observations: 
\begin{enumerate} 
\item $|S_t|$ will keep decreasing until it contains $O(\log n)$ many non-$\PC$ vertices. Most of the vertices removed are non-$\PC$ vertices, so at the end, a $1-o(1)$ fraction of $S_t$ is in $\PC$.
\item If $S_t$ is such that most of its vertices are members of $\PC$, the gradient descent algorithm will converge to $\PC$ in a further $O(k)$ steps. 
\end{enumerate}
These are formalized by the following two lemmas. 

\begin{Lemma}\label{lemma:monotone}
Let $\gamma>3$. 
There exists an absolute constant $c_0(\gamma)>0$ such that with probability $1-o(1)$, for all $U\subseteq [n]$ with $|U\cap [k]^c|\geq c_0\log n$, 
\begin{align}
\min_{x\in U}H(U\setminus\{x\})\leq\min\left\{H(U), \min_{z\in U^c}H(U\cup\{z\})\right\} - 1. \label{happy}
\end{align}
Consequently, $(|S_t|)_{t\ge 0}$ will only decrease until $|S_t\cap [k]^c|\leq c_0\log n$. 
\end{Lemma}

\begin{Lemma}\label{lemma:never}
    Let $\gamma>3$ and $0<\xi<1-\frac{1}{2}(1+\frac{1}{\gamma})$. With probability $1-o(1)$ the following holds. For any $S$ having $|S\cap [k]| \ge \max\{\gamma |S\cap [k]^c| + 2, (1-\xi)k\}$, if $S_t'$ is the gradient descent initialized from $S_0'=S$, then $d_\H(S_{t}',\PC)$ is strictly decreasing in $t$ while $S_t'\ne \PC$, where $d_\H$ denotes the Hamming distance. Consequently $S_t'$ will converge to $\PC$ in at most $2k$ steps. 
    
    Moreover, for any $W\subseteq [n]$, $W\ne \PC$ that satisfies $|W\cap [k]|\geq \max\{\gamma|W\cap [k]^c|+2, (1-\xi)k\}$, 
\begin{align}
\min_{U\in\mathcal U} H(U)\leq \min_{U\sim W, U\notin\mathcal U}H(U) - 1, \label{egs}
\end{align} 
where $\mathcal U = \left\{U\sim W: d_\H(U, \PC)<d_\H(W, \PC)\right\}$ represents the set of neighboring states of $W$ that are one Hamming distance closer to the $\PC$ than $W$ is.  
\end{Lemma}


The proofs of Lemma \ref{lemma:monotone} and Lemma \ref{lemma:never} are based on the degree concentration estimates of Lemma~\ref{lemma:000} and are deferred to Section \ref{sec:proofs}. 
To stitch the two lemmas together, we introduce a \emph{peeling process} $Y_t$, initialized from $S_0 = V$, that at each time $t\geq 1$ removes the vertex in $Y_{t-1}$ having the smallest degree in $Y_{t-1}$ (if there are multiple, pick randomly) to obtain $Y_t$. 
By the form of the Hamiltonian~\eqref{myH},
\begin{align}
    \argmin_{x\in U} H(U\setminus \{x\}) = \argmin_{x\in U} |E(x,U)|\label{xstar},
\end{align}
(where the $\argmin$'s are understood as sets when they are not singletons). Therefore, by Lemma~\ref{lemma:monotone}, with high probability,
$Y_t = S_t$ until the first time $|Y_t\cap [k]^c|\le c_0 \log n$. 
For convenience, let $\tau_0=\tau_0(c_0)$ be the hitting time for this event.  
The process $Y_t$ turns out to coincide exactly with the removal stage of the algorithm of~\cite{feige2010finding} (though now stopped at $\tau_0$).
The following lemma says that with large probability, a large portion of $[k]$ remains well-connected in $Y_t$. In particular, we define the random subset of vertices of $\PC$ that retain degrees close to their expectations (which is larger than those of typical non-$\PC$ vertices) throughout the removal process: for any $c_1>0$ let 
\begin{align}
\mathcal A(c_1) = \left\{x\in [k]: \deg(x, Y_t)\geq (|Y_t\cap [k]|-1) + \frac{1}{2}|Y_t\cap [k]^c| - c_1\sqrt{n}\ \text{ for all $t<T_x$}\right\}.\label{mya}
\end{align} 
The following lemma shows that most of $\PC$ is in this set.

\begin{Lemma}\label{lemma:A}
Let $Y_t$ denote the peeling process defined above. 
For $x\in [n]$, let $T_x$ denote the time that $x$ is removed from $Y_t$, i.e., $T_x = \min\{t: x\notin Y_t\}\wedge \tau_0$. 
For any $\e, \eta>0$, there exists an absolute constant $c_1(\e, \eta)>0$ such that with probability at least $1-\e$, $\left|\mathcal A\right|\geq (1-\eta)k$.
\end{Lemma}
The proof is similar to \cite[Corollary 5]{feige2010finding}; for completeness, we reproduce the proof in a more general setting when discussing the robustness of the algorithm in Section \ref{sec:rob}, namely Lemma \ref{lemma:A+v}. 

\begin{proof}[\textbf{\emph{Proof of Theorem \ref{mainthm:success}}}]
We first prove the result for the gradient descent $S_t$ and then extend it to the positive temperature chain $S_t^\beta$ via a coupling argument. 
According to Lemma \ref{lemma:monotone} and~\eqref{xstar}, there exists $c_0>0$ such that, with probability $1-o(1)$, there exists a valid coupling of $S_t$ and $Y_t$ such that they coincide up to time $\tau_0(c_0)$, resulting in $S_{\tau_0} = Y_{\tau_0}$. 
For any $\eta<1-\frac{1}{2}(1+\frac{1}{\gamma})$, we claim that there exists $c_1(\e,\eta)>0$ such that with probability at least $1-\e$, 
\begin{align}\label{nts-main-proof}
|\A|\geq (1-\eta)k \quad \text{and} \quad \A\subseteq Y_t \quad \text{for all $t\leq\tau_0$},
\end{align}
where $c_1$ appears in the definition of $\A$ in \eqref{mya}. 
Assuming that, then for $t = \tau_0$, 
\begin{align*}
\gamma |S_{\tau_0}\cap [k]^c|+2 = \gamma c_0\log n + 2\leq (1-\eta)k\stackrel{\eqref{nts-main-proof}}{\leq} |\A|\leq |Y_{\tau_0}\cap [k]|=|S_{\tau_0}\cap [k]|.
\end{align*}
The proof of the gradient descent part is completed by appealing to Lemma \ref{lemma:never} with $S_0'= S_{\tau_0}$ and noting that the total number of steps to reach $S_t = \PC$ is at most $\tau_0 + 2k\leq n+2k$. Note that once the gradient descent has reached $\PC$, it will be absorbed per Theorem~\ref{thm:global optimum}. 

To verify~\eqref{nts-main-proof}, we apply an inductive argument.  
Let $k \geq C\sqrt{n}$ for some absolute constant $C$ to be determined during the proof.
Taking a union bound for the statements in Lemmas \ref{lemma:000}, \ref{lemma:000+}, \ref{lemma:A}, we see that the following events hold simultaneously with probability at least $1-\e$: there exist $c, c_1, \rho>0$ such that
\begin{align}
\min_{x\in U'}|E(x, U)|&\leq\left(1+\delta\right)\frac{\E[\deg(U', U)]}{|U'|}&\forall U'\subseteq U\subseteq [n], |U'|\geq c\log n\label{t1}\\
\min_{x\in U\cap [k]^c}|E(x, U)|&\leq \frac{1}{2}(|U|-1) + \rho\sqrt{n}&\forall U\subseteq [n], |U\cap [k]^c|\geq\frac{n}{\log n}\label{t2}\\
\left|\mathcal A\right| = \left|\mathcal A(c_1)\right|&\geq (1-\eta)k\label{t3},
\end{align}
where $\delta = \delta(U') = \sqrt{\frac{96\log n}{|U'|}}$. 
It is easy to see that $\A\subseteq Y_0 = [n]$ at $t=0$.
To apply induction, we assume $\A\subseteq Y_{t-1}$ and verify $\A\subseteq Y_{t}$. 
For ease of presentation, we let $|Y_t| = n_{t, 1} + n_{t, 2}$ where $n_{t,1}=|Y_{t}\cap [k]|, n_{t,2} = |Y_{t}\cap [k]^c|$.
If $n_{t-1,2}\geq \frac{n}{\log n}$, letting $U = Y_{t-1}$ in \eqref{t2}, we have
\begin{align*}
\min_{y\in Y_{t-1}\cap [k]^c}& \deg(y, Y_{t-1})  \stackrel{\eqref{t2}}{\leq} \frac{1}{2}(|Y_{t-1}|-1) + \rho\sqrt{n}\\
\leq&\ (n_{t-1, 1}-1) + \frac{1}{2}n_{t-1, 2} - c_1\sqrt{n} - \left[\frac{1}{2}(n_{t-1, 1}-1) - \rho\sqrt{n} - c_1\sqrt{n}\right] \\
\stackrel{\eqref{mya}}{\leq}&\ \min_{x\in \A\subset Y_{t-1}}\deg(x, Y_{t-1}) -  \left[\frac{1}{2}(n_{t-1, 1}-1) - \rho\sqrt{n} - c_1\sqrt{n}\right]\\
\stackrel{\eqref{t3}}{\leq}&\ \min_{x\in \A\subset Y_{t-1}}\deg(x, Y_{t-1}) -  \left[\frac{C(1-\eta)}{2}\sqrt{n} - \rho\sqrt{n} - c_1\sqrt{n}\right]< \min_{x\in \A\subset Y_{t-1}}\deg(x, Y_{t-1}),
\end{align*}
where the last inequality holds for any $C$ satisfying $C>\frac{2(\rho+ c_1)}{1-\eta}$. 
If $c \log n \le n_{t-1,2}\leq \frac{n}{\log n}$, letting $(U, U') = (Y_{t-1}, Y_{t-1}\cap [k]^c)$ in \eqref{t1}, we have 
\begin{align*}
\min_{y\in Y_{t-1}\cap [k]^c} & \deg(y, Y_{t-1}) \stackrel{\eqref{t1}}{\leq} \frac{1}{2}(1+\delta)(|Y_{t-1}|-1)\\
=&\ (n_{t-1, 1}-1) + \frac{1}{2}n_{t-1, 2} - c_1\sqrt{n} - \left[\left(1-\frac{1+\delta}{2}\right)(n_{t-1, 1}-1) -\frac{\delta}{2} n_{t-1, 2} - c_1\sqrt{n} \right] \\
\stackrel{\eqref{mya}}{\leq}&\ \min_{x\in \A\subset Y_{t-1}}\deg(x, Y_{t-1}) - \left[\left(1-\frac{1+\delta}{2}\right)(n_{t-1, 1}-1) -\frac{\delta}{2} n_{t-1, 2} - c_1\sqrt{n} \right]\\
\stackrel{\eqref{t3}}{\leq}&\ \min_{x\in \A\subset Y_{t-1}}\deg(x, Y_{t-1}) - \left[\frac{C(1-\eta)}{4}\sqrt{n} - \sqrt{24n} - c_1\sqrt{n} \right]<\min_{x\in \A\subset Y_{t-1}}\deg(x, Y_{t-1}),
\end{align*}
where the last inequality holds if $C>\frac{4(\sqrt{24}+c_1)}{1-\eta}$. 
Combining the two cases together, we have $\A\subseteq Y_t$.

To extend the above result to the positive temperature chain when $\beta = \Omega(\log n)$, it suffices to show that the following events hold with high probability:
\begin{itemize}
\item $S_t$ and $S^\beta_t$ remain fully coupled before $\tau_0$, i.e. $S_{\tau_0}=S^\beta_{\tau_0}$;
\item For $t>\tau_0$, the Hamming distance between $S^\beta_t$ and $\PC$ is decreasing before reaching $\PC$;
\item Once it has reached $\PC$, $S_t^\beta$ will stay there for $n^{\Theta(k)}$ steps. 
\end{itemize} 
Indeed, for $t<\tau_0$, denote $\mathcal M_t = \argmin_{U\sim S_t}H(U)$ and $m_t = \min_{U\sim S_t}H(U)$.
Suppose that $S_t$ and $S_t^\beta$ are equal at time $t$. 
In this case, $S_t$ and $S_t^\beta$ will remain the same at time $t+1$ if the positive temperature chain moves to a neighboring state of $S_t$ in $\mathcal M_t$ (if $\mathcal M_t$ is not a singleton, the choice of which one can be coupled trivially).  
The probability of this event, according to \eqref{glauber}, is at least
\begin{align}\label{eq:glauber-transition-bound}
\frac{|\mathcal M_t| e^{-\beta m_t}}{|\mathcal M_t| e^{-\beta m_t} + \sum_{U\sim S_t, U\notin\mathcal M_t}e^{-\beta H(U)}}\geq\frac{1}{1+ ne^{-\beta\Delta}}>1-ne^{-\beta\Delta},
\end{align}
where $\Delta:=\min_{U\sim S_t, U\notin\mathcal M_t}H(U)-m_t$ is the energy gap between the smallest and second to smallest energy in the neighboring states of $S_t$. 
Meanwhile, it follows from Lemma \ref{lemma:monotone} that $\mathcal M_t\subseteq \{U\sim S_t: |U|< |S_t|\}$, and consequently, 
\begin{align*}
\Delta = \min\left\{\min_{U\sim S_t, U\notin\mathcal M_t, |U|\geq |S_t|}H(U)-m_t, \min_{U\sim S_t, U\notin\mathcal M_t, |U|< |S_t|}H(U)-m_t\right\}\stackrel{\eqref{pl2}, \eqref{happy}}{\geq} 1.
\end{align*}
Substituting this into the above estimate we obtain that $S^\beta_{t+1}=S_{t+1}$ holds with probability at least $1-ne^{-\beta\Delta}\geq 1- n^{-2}$ if $\beta = \Omega(\log n)$.  
Applying a union bound over $t\leq\tau_0$ and noting $\tau_0<n$ concludes that with probability at least $1-n^{-1}$, $S^\beta_{\tau_0}=S_{\tau_0}$.  

For $t \geq \tau_0$, denote the neighboring states of $S_t^\beta$ that are one Hamming distance closer to $\PC$ than $S_t^\beta$ as $\mathcal U^\beta_{t} = \{U\sim S^\beta_t, d_\H(U, \PC)<d_\H(U, S^\beta_t)\}$.
When $t = \tau_0$, $\mathcal U^\beta_{\tau_0} = \mathcal U_{\tau_0}$ and $S^\beta_{\tau_0}=S_{\tau_0}$ satisfies the condition on the initialization in Lemma \ref{lemma:never}. 
By a similar calculation to~\eqref{eq:glauber-transition-bound}, using \eqref{egs}, we have $S^\beta_{t+1}\in\mathcal U^\beta_{t}$ with probability at least $1-n^{-2}$.
In this case, as opposed to the situation before $\tau_0$, the energy gap (the difference between the smallest neighboring energy and the second lowest neighboring energy) amongst elements of $\mathcal U_t$ may not be lower bounded by an absolute positive constant, so we may not expect $S^\beta_{t+1}$ and $S_{t+1}$ to be fully coupled in this stage. Nevertheless, both chains are moving in the right direction in the sense that their Hamming distance to $\PC$ is strictly decreasing. 
In particular, $d_\H(S^\beta_{t+1}, \PC)<d_\H(S^\beta_{t}, \PC)$ and $S^\beta_{t+1}$ can be analyzed similarly as in the previous step.  
By repeating this process, we can show via a union bound that with probability at least $1-|\mathcal U^\beta_{t}|\cdot n^{-2}\geq 1-n^{-1}$, $d_\H(S^\beta_{\ell+1}, \PC)<d_\H(S^\beta_{\ell}, \PC)$ for all $\ell\geq\tau_0$ until $d_\H(S^\beta_{\ell}, \PC) = 0$. 
Consequently, $S^\beta_{t} = \PC$ with $t = \tau_0 + |\mathcal U^\beta_{\tau_0}|$. 
Combining the portion before $\tau_0$ and the portion after, we establish that the positive temperature chain takes $n+2k$ steps to find $\PC$ with probability at least $1-o(1)$ if $\beta =\Omega(\log n)$.
Meanwhile, there exists an absolute constant $c_2>0$ such that
\begin{align}\label{eq:min-gap-PC}
\min_{U\sim [k], U\neq [k]}H(U) - H([k])\stackrel{\eqref{obs}, \eqref{obs1}}{\geq} c_2k. 
\end{align}
Thus, upon reaching $\PC$, the probability of staying in $\PC$ in the next move is at least $1-e^{-(c_2k-1)\log n}$.
Thus, it will stay in $\PC$ for at least $e^{\frac{(c_2k-1)}{2}\log n}\geq n^{\frac{k}{C}}$ ($C>3/c_2$) steps with probability tending to one. 
\end{proof}

\section{Failure from the empty-set initialization}\label{sec:empty-init}

In this section, we prove Theorem~\ref{mainthm:failure}, showing that if the initialization were $S_0 = \emptyset$ instead of $S_0 = V$, then the gradient descent and positive temperature chain processes would fail whenever the planted clique has any sub-linear size. 
This matches the failure result of~\cite{chen2023almost} and demonstrates that while both $\emptyset$ and $V$ are natural uninformed initializations one could hope the MCMC succeeds from, it is crucial that the latter choice be made in this problem. 

We start with some intuition as to why starting from the empty set does not work, even with the relaxed Hamiltonian of~\eqref{myH}.
In its initial stages, started from $S_0 = \emptyset$, $S_t$ increases while remaining a clique until its size reaches $\Theta(\log n)$ and may not have any lower-energy neighbors that are cliques.
In this period, it is simply a greedy algorithm moving in the space of cliques, and the relaxation to non-cliques plays no role. 
When the gradient descent has reached the $O(\log n)$-size and starts to move off of cliques, it is already close to some near-clique local minimum with no overlap with $\PC$, and gets absorbed into that state.  

To make this a rigorous proof, we consider a coupled \ER random graph $G_0 \sim \G(n, \frac{1}{2})$ with $G$, where one first generates $G_0 \sim \G(n, \frac{1}{2})$ and then completes the missing edges between vertices of $[k]$ to obtain $G\sim \G(n, \frac{1}{2}, k)$. 
We can then run coupled gradient descent processes on both $G_0$ and $G$, denoting the corresponding processes by $\widetilde{S}_t$ and $S_t$.
The trajectories of $\widetilde{S}_t$ and $S_t$ can be fully coupled up to time $\tau$ where $\tau$ is the first time $S_i\cap [k] \ne \emptyset$. 
This is seen by noting that while $\widetilde{S}_i = S_i$, the energies of all possible transitions are identical, as they only depend on edges incident to $S_i$ (and not any internal edges of $\PC$). Thus the next steps of the processes are coupled identically.  
This observation leads to the following.

\begin{Lemma}\label{lm:change}
Let $S_t$ and $\widetilde{S}_t$ be the two Markov chains running on $G, G_0$ that are coupled in the way described above, and $\tau$ be the first time that $S_t$ intersects with $\PC$. 
For any $\gamma_1>0$ and $L = (\log n)^{\gamma_1}$, $\P(\tau>L) = 1-o(1)$. 
Consequently, $S_t=\widetilde{S}_t$ for all $t\leq L$ with probability $1-o(1)$. 
\end{Lemma}

That is to say, if the un-planted model's gradient descent is absorbed in $\text{poly}(\log n)$ time, then the two chains $S_t, \widetilde S_t$ are perfectly coupled and $ S_t$ is also absorbed.
This suggests studying the dynamics of $\widetilde{S}_t$ as a proxy for $S_t$ as long as it terminates in $\text{poly}(\log n)$ time. The next lemma verifies this statement.

\begin{Lemma}\label{lm:changee}
Let $\widetilde{S}_t$ be the gradient descent running on $G_0 \sim \G(n, \frac{1}{2})$. 
Denote the absorption time of $\widetilde{S}_t$ by $\widetilde{T}$. 
There exists  $\gamma_1>0$ such that with probability $1-o(1)$, $\widetilde{T}\leq (\log n)^{\gamma_1}$. 
Moreover, the terminal state $\widetilde S_\infty$ has no overlap with $\PC$, i.e. $\widetilde S_\infty \cap [k] = \emptyset$.
\end{Lemma}
The proof of Lemma \ref{lm:changee} follows by noting that the global minima in the \ER must have polylog energy at most, and that in a constant fraction of its steps before $\widetilde T$, $H_{G_0}(S_t)$ decreases by an $\Omega(1)$ amount. This latter step is actually somewhat delicate, since for general $\gamma$, the energy gaps in $H$ do not have uniform ($n$ independent) lower bounds. The proofs of both Lemmas \ref{lm:change} and \ref{lm:changee} are given in Section \ref{sec:proofs}.

\begin{proof}[\textbf{\emph{Proof of Theorem \ref{mainthm:failure}}}]
By Lemma~\ref{lm:changee}, with probability $1-o(1)$, the process $\widetilde S_t$ is absorbed before time $L$ from Lemma~\ref{lm:change}. On that event, by Lemma~\ref{lm:change}, $\tau >L$ and $S_{\widetilde T} = \widetilde S_{\widetilde T}$, and finally this is also an absorbing state for $S_{\widetilde T}$ since while $S_{\widetilde T} \cap [k] = \emptyset$, under the coupling of $(G_0,G)$, its energy and those of all its neighbors on the hypercube are identical to those for $\widetilde S_{\widetilde T}$.  
\end{proof}

\section{Deferred proofs}\label{sec:proofs}
In this section, we include the technical proofs that were deferred from the above sections. 
\subsection{Degree concentration}
\begin{proof}[\textbf{\emph{Proof of Lemma \ref{lemma:000}}}]
The proof follows from a union bound argument. Write $\deg(U', U)$ as 
\begin{align*}
\deg(U', U) = \sum_{x\in U', z\in U}\mathbb I\{(x, z)\in E\} = \sum_{x\in U', z\in U\setminus U'}\mathbb I\{(x, z)\in E\} + 2\sum_{x, z\in U', z\neq x}\mathbb I\{(x, z)\in E\},
\end{align*}
which is a sum of $n_{U'}:=|U'|(|U|-|U'|) + {|U'|\choose 2}$ independent random variables bounded by $2$. 
By the Chernoff bound, for $\delta\in (0,1)$, 
\begin{align*}
\P\left(|\deg(U', U)-\E[\deg(U', U)]|\geq \delta\E[\deg(U', U)]\right)\leq 2e^{-\frac{\delta^2\E[\deg(U', U)]}{8}}\leq 2e^{-\frac{\delta^2n_{U'}}{16}},
\end{align*}
where the last inequality follows from $\E[\deg(U', U)]\geq \frac{n_{U'}}{2}$. 
Taking $\delta\geq\sqrt{\frac{96\log n}{|U'|}}$ yields 
\begin{align}
\P\left(|\deg(U', U)-\E[\deg(U', U)]|\geq \delta\E[\deg(U', U)]\right)\leq 2n^{-3|U|}. \label{bad}
\end{align}
Taking a union bound of \eqref{bad} shows that $|\deg(U', U)-\E[\deg(U', U)]|\geq \delta\E[\deg(U', U)]$ holds simultaneously for all $U'\subseteq U\subseteq [n]$ with $|U'|\geq c\log n$ with probability at least
\begin{align*}
\sum_{U\subseteq [n]: |U|\geq c\log n}\sum_{U'\subseteq U: c\log n\leq |U'|\leq |U|}2n^{-3|U|}\leq \sum_{r\geq c\log n}{n\choose r}\cdot 2^rn^{-3r}\,.
\end{align*}
This is at most $n^{-\alpha}$ for any $\alpha>0$ by taking $c$ large enough.
In particular, this proves the desired bound of~\eqref{eq:degree-concentration}. The second part of the lemma follows immediately from the fact that the minimum degree has to be at most the average degree of $U'$ in $U$.  
\end{proof}

\begin{proof}[\textbf{\emph{Proof of Lemma \ref{lemma:000+}}}]
The proof is similar to the proof of Lemma \ref{lemma:000}.  
For fixed $U$, choosing $U' = U\cap [k]^c$ as in the previous proof and applying the Chernoff bound, we have for $\delta\in (0, 1)$, 
\begin{align*}
\P\left(|\deg(U', U)-\E[\deg(U', U)]|\geq \delta\E[\deg(U', U)]\right)\leq 2e^{-\frac{\delta^2\E[\deg(U', U)]}{8}}\leq 2e^{-\frac{\delta^2n_{U'}}{16}},
\end{align*}
where $n_{U'}:=|U'|(|U|-|U'|) + {|U'|\choose 2}$. 
To obtain a bound for all $U$ with $|U\cap [k]^c|\geq\frac{n}{\log n}$ (the total number of such $U$ is bounded by $2^n$) with probability at least $1-n^{-2}$, we require $2^{n}\cdot 2e^{-\frac{\delta^2n_{U'}}{16}} = e^{-\frac{\delta^2n_{U'}}{16} + (n+1)\log 2}\leq n^{-2}$, which holds if we choose $\delta = \sqrt{\frac{16(2\log n+(n+1)\log 2)}{n_{U'}}}$.
In this case, 
\begin{align*}
\min_{x\in U'}|E(x, U)|&\leq\frac{\deg(U', U)}{|U'|}\leq (1+\delta)\frac{\E[\deg(U', U)]}{|U'|}\\
& = \frac{1}{2}(|U|-1)+\sqrt{\frac{16(2\log n+(n+1)\log 2)(|U|-1)^2}{n_{U'}}}\\
&\leq \frac{1}{2}(|U|-1)+\rho\sqrt{n},
\end{align*}
where the last step holds for some constant $\rho>0$ since 
\begin{align*}
(|U|-1)^2 \leq (|U'|+k)^2\stackrel{|U'|\geq \frac{n}{\log n}> k}{\leq} 4|U'|^2  = O( {|U'|\choose 2})\,,
\end{align*}
and, by definition, $n_{U'}\ge \binom{|U'|}{2}$.
\end{proof}

\subsection{Deferred proofs for success from full initialization}

\begin{proof}[\textbf{\emph{Proof of Lemma \ref{lemma:monotone}}}]
Choose $c_0$ at least as large as the constant $c$ in Lemma \ref{lemma:000}, and furthermore large enough that 
\begin{align}
\omega:=\frac{\left(1+\delta\right)\left(\gamma + 1\right)}{2(\gamma - 1)}\leq \frac{\left(1+\sqrt{\frac{96}{c}}\right)\left(\gamma + 1\right)}{2(\gamma - 1)}<1. \label{goodone}
\end{align} 
Note that such a $c$ always exists for $\gamma>3$. 
By Lemma \ref{lemma:000}, specifically~\eqref{panger} applied with $U = U$ and $U' = U \cap [k]^c$, there exists $x\in U\cap [k]^c$ such that 
\begin{align}
H(U\setminus\{x\}) - H(U)\stackrel{\eqref{pl2}}{\leq} \frac{1}{2}\left(1+\delta\right)\left(\gamma + 1\right)(|U|-1)-\gamma(|U|-1)\stackrel{\eqref{goodone}}{\leq}-(|U|-1)< -1. \label{lkhs}
\end{align}
For the same $x$ and any $z\in U^c$, 
\begin{align}
H(U\setminus\{x\})-H(U\cup\{z\})&\stackrel{\eqref{pl1}, \eqref{lkhs}}{\leq} \left[-\gamma + 1 + \frac{1}{2}(1+\delta)(\gamma + 1)\right](|U|-1) + \gamma\nonumber\\
&\leq -(1-\omega)(|U|-1)+\gamma\nonumber\\
&\leq -(1-\omega)(c_0\log n-1)+\gamma<-1. \label{liop}
\end{align}
Combining \eqref{lkhs} and \eqref{liop} finishes the proof. 
\end{proof}

\begin{proof}[\textbf{\emph{Proof of Lemma \ref{lemma:never}}}]
Assume the event of Lemma \ref{lemma:monotone} holds for some $c_0>0$. Recall that $S_t'$ is a realization of the gradient descent chain initialized from $S_0'=S$. 
If $|S_t' \cap [k]^c|>c_0\log n$, then $|S_{t+1}'|<|S_t'|$. 
Meanwhile, a similar calculation as \eqref{obs} shows 
\begin{align}
&H(S_{t}'\setminus\{x\})-H(S_t') \geq (|S_t'\cap [k]|-1) - \gamma |S_t'\cap [k]^c|\geq 1&\forall x\in S_t'\cap [k], \label{eg1}
\end{align}
i.e., none of the vertices in $S_t' \cap [k]$ will be removed in the first step. 
Therefore, the removed vertex is from $S_t'\cap [k]^c$, which implies $d_\H(S_{t+1}', \PC)<d_\H(S_t', \PC)$. 

If $|S_t'\cap [k]^c|\leq c_0\log n$, \eqref{eg1} still holds. 
Meanwhile, under the assumption $|S_t'\cap [k]|\geq (1-\xi)k$ with $\xi<1-\frac{1}{2}(1+\frac{1}{\gamma})$, letting $\zeta = \frac{1-\xi - \frac{1}{2}(1+\frac{1}{\gamma})}{2(1+\frac{1}{\gamma})}>0$, 
we have for any $x\in (S_t')^c\cap [k]^c$, 
\begin{align}
H(S_t'\cup\{x\}) - H(S_t') &= \gamma |S_t'| - (1+\gamma)|E(x, S_t')| \geq \gamma |S_t'\cap [k]| - (1+\gamma)|E(x, [k])|- |E(x, S_t'\cap [k]^c)|\nonumber\\
&\geq \gamma(1-\xi) k - (1+\gamma)\left(\frac{1}{2}+\zeta\right)k - c_0\log n\nonumber\\
& \geq \frac{\gamma(1-\xi)}{2}k - c_0\log n>1,\label{eg2}
\end{align}
where the second inequality holds with high probability $1-o(1)$ as a result of standard concentration bounds for $|E(x, [k])|$ for all $x\in [k]^c$, and the last step uses $k>\frac{2c_0}{\gamma (1-\xi)}\log n$. 
This shows $|S_t'\cap [k]^c|$ is nonincreasing.
Since $S_t'$ is in a regime where $H$ has no local minima (Remark~\ref{rem:no-complexity-above}), either $|S_t'\cap [k]^c|$ decreases or $S_t'\cap [k]$ increases. 
In particular, $d_\H(S_{t+1}', \PC)<d_\H(S_{t}', \PC)$. 
The proof is finished by noting such a process can last at most $d_\H(S_t', \PC)$ steps, which can be bounded as
\begin{align*}
d_\H(S_t', \PC) = |S_t'\cap [k]^c|+|(S_t')^c\cap [k]|\leq \frac{1-\xi}{\gamma}|S_t'\cap [k]| + k\leq 2k.
\end{align*}
The corresponding energy gap estimate of \eqref{egs} follows by combining \eqref{eg1} and \eqref{eg2}.  
\end{proof}

\subsection{Deferred proofs for failure from empty initialization}
\begin{proof}[\textbf{\emph{Proof of Lemma \ref{lm:change}}}]
Assuming $t<\tau$ so that $S_t = \widetilde S_t$ and $S_t\cap [k] = \emptyset$, then the probability of $S_{t+1}\cap [k]\neq\emptyset$ is equal to the probability that a vertex is added and that vertex is from $[k]$. In particular, the probability that $\tau\le L$ is bounded by the probability that in one of the first $L$ steps, $\widetilde S_t$ adds a vertex in $[k]$.
By the exchangeability of vertices in $G_0$, this probability is upper bounded by $\frac{k}{|S_t^c|}\leq \frac{k}{n-L}$. 
The desired result follows by taking a union bound over $t\leq L$. 
\end{proof}

\begin{proof}[\textbf{\emph{Proof of Lemma \ref{lm:changee}}}]
Throughout the proof, we use notation $E_{G_0}$, $H_{G_0}$, etc., to emphasize that we are working under the \ER model $G_0$. 
We start by noting that Lemma \ref{lemma:000} does not depend on $k$ so it also applies to $G_0$. 
As a result, there exists $c(\gamma)>0$ such that with probability at least $1-o(1)$, for all $U\subseteq [n]$ with $|U|\geq c\log n$ and $\delta = \sqrt{\frac{96\log n}{|U|}}\leq\sqrt{\frac{96}{c}}<\frac{\gamma -1}{\gamma+1}$, 
\begin{align*}
|E_{G_0}(U)| = \frac{1}{2}\deg_{G_0}(U, U)\leq (1+\delta)\E[\deg_{G_0}(U, U)]<\frac{\gamma}{1+\gamma}{|U|\choose 2}.
\end{align*}
The Hamiltonian of such $U$'s are lower bounded by 
\begin{align*}
&H_{G_0}(U) = \gamma\left[{|U|\choose 2}-\left(1+\frac{1}{\gamma}\right)|E_{G_0}(U)|\right]>0&|U|\geq c\log n. 
\end{align*}
Since $H_{G_0}(\widetilde{S}_t)$ is strictly decreasing prior to $\widetilde T$, and $H_{G_0}(\widetilde{S}_0) = 0$, $|\widetilde{S}_t|\leq c\log n$ for all $t\leq \widetilde{T}$. 
Consequently, 
\begin{align}
\min_{t\leq \widetilde{T}}H_{G_0}(\widetilde{S}_t)\geq -{c\log n\choose 2}\geq -\frac{(c\log n)^2}{2}.\label{ak}
\end{align}
Meanwhile, we claim that $|\{t<\widetilde{T}, H_{G_0}(\widetilde{S}_{t+1}) - H_{G_0}(\widetilde{S}_t)\leq -\frac{1}{2}\}|\geq\frac{\widetilde{T}}{2}$.  
This combined with \eqref{ak} shows that $\widetilde{T}\leq 2(c\log n)^2$. 
In particular, since $\widetilde T\le L$ (for large choice of $\gamma_1$), the probability that the absorbed state $\widetilde{S}_{\widetilde{T}}$ has no intersection with $\PC$ is at least $1-o(1)$. 

It remains to verify the claim to finish the proof.
Without loss of generality, we first assume that $\gamma$ is a rational number, i.e., $\gamma = \frac{q_1}{q_2}$ for some $q_1>q_2$ and $q_2$ is a prime number, and then extend to the general case of real numbers via a continuity argument. 
For any $t<\widetilde{T}-1$, suppose that $\widetilde{S}_{t+1}$ is obtained from $\widetilde{S}_{t}$ ($|\widetilde{S}_t| = m$) by adding a vertex $x\in \widetilde{S}_{t}^c$, i.e. $\widetilde{S}_{t+1} = \widetilde{S}_{t}\cup\{x\}$.
By definition of gradient descent, 
\begin{align*}
H_{G_0}(\widetilde{S}_{t}\cup\{x\})-H_{G_0}(\widetilde{S}_{t})<0\stackrel{\eqref{pl1}}{\Longrightarrow} |E_{G_0}(x, \widetilde{S}_{t})|>\frac{\gamma m}{1+\gamma}.
\end{align*}
If $\frac{\gamma m}{1+\gamma}$ is an integer, then $|E_{G_0}(x, \widetilde{S}_{t})|\geq\frac{\gamma m}{1+\gamma}+1$. 
In this case, it is easy to check $H_{G_0}(\widetilde{S}_{t}\cup\{x\})-H_{G_0}(\widetilde{S}_{t})<-(1+\gamma)<-1$. 
When $\frac{\gamma m}{1+\gamma}$ is not an integer, we can represent it as the difference between its ceiling part and the decimals:
\begin{align*}
&\frac{\gamma m}{1+\gamma} = \frac{q_1 m}{q_1+q_2} = s - \frac{r}{q_1+q_2}& r<q_1+q_2.
\end{align*}
In this case, we have 
\begin{align*}
&H_{G_0}(\widetilde{S}_{t}\cup\{x\})-H_{G_0}(\widetilde{S}_{t}) \leq -(1+\gamma)\left(\left\lceil\frac{\gamma m}{\gamma+1}\right\rceil - \frac{\gamma m}{\gamma+1}\right)\\=&\ -\frac{q_1+q_2}{q_2}\left(\left\lceil\frac{q_1 m}{q_1+q_2}\right\rceil - \frac{q_1 m}{q_1+q_2}\right)= -\frac{r}{q_2}.
\end{align*}
A key observation is the following. If $r\geq \frac{q_2}{2}$, then $H_{G_0}(\widetilde{S}_{t}\cup\{x\})-H_{G_0}(\widetilde{S}_{t})<-\frac{1}{2}$. 
Otherwise, a large energy decrease must occur in the next step, i.e., $H_{G_0}(\widetilde{S}_{t+2})-H_{G_0}(\widetilde{S}_{t+1})<-1$.
To see this, note there are two cases that might be happening at $t+2$:
\begin{itemize}
\item If $\widetilde{S}_{t+2} = \widetilde{S}_{t+1}\cup\{z\}$ for some $z\in \widetilde{S}^c_{t+1}$, by a similar computation to the above, 
\begin{align*}
H_{G_0}(\widetilde{S}_{t+1}\cup\{z\})-H_{G_0}(\widetilde{S}_{t+1}) &\leq -\frac{q_1+q_2}{q_2}\left(\left\lceil\frac{q_1 (m+1)}{q_1+q_2}\right\rceil - \frac{q_1 (m+1)}{q_1+q_2}\right).
\end{align*}
Note 
\begin{align*}
&\frac{q_1 (m+1)}{q_1+q_2} = s - \frac{r}{q_1+q_2} + \frac{q_1}{q_1+q_2} = (s+1) - \frac{q_2+r}{q_1+q_2},
\end{align*}
where the subtracted fraction in the last equality is strictly less than $1$ since $r<\frac{q_2}{2}<q_1$. 
Consequently, 
\begin{align*}
    H_{G_0}(\widetilde{S}_{t+1}\cup \{z\}) - H_{G_0}(\widetilde{S}_{t+1}) \le - \frac{q_1 + q_2}{q_2}\left(\frac{q_2 + r}{q_1+q_2}\right) = - \frac{q_2 +r}{q_2} \le -1.
\end{align*}
\item If $\widetilde{S}_{t+2} = \widetilde{S}_{t+1}\setminus\{z\}$ for some $z\in \widetilde{S}_{t+1}$, then
\begin{align*}
H_{G_0}(\widetilde{S}_{t+1}\setminus\{z\})-H_{G_0}(\widetilde{S}_{t+1}) &\stackrel{\eqref{pl2}}{\leq} -\frac{q_1+q_2}{q_2}\left(\frac{q_1 m}{q_1+q_2} - \left\lfloor\frac{q_1 m}{q_1+q_2}\right\rfloor\right) = -\frac{q_1+q_2-r}{q_2}<-1.
\end{align*}
\end{itemize}
The case where $\widetilde{S}_{t+1} = \widetilde{S}_{t}\setminus\{x\}$ for some $x\in\widetilde{S}_{t}$ can be reasoned similarly. 
Putting the above discussions together, we conclude that either the energy decrease from $t$ to $t+1$ is at least $-\frac{1}{2}$, or the energy decrease from $t+1$ to $t+2$ is at least $-\frac{1}{2}$. 
This implies that in at least half of the steps of the gradient descent prior to termination, there is an energy decrease of $-\frac{1}{2}$ (uniformly over all rational $\gamma$). For any $\gamma$ (possibly irrational), there is a rational $\gamma'$ sufficiently close to $\gamma$ such that the gradient descent moves are all identical for at least $n$ steps---note that this $\gamma'$, and in particular its denominator, will depend on $n$. But the uniformity of the above estimate over the denominator in $\gamma'$ implies that also for the  $\gamma$-dynamics in the first $n$ steps, half of the steps prior to absorption lower the energy by at least $\frac{1}{2}$. 
\end{proof}

\subsection{Proof of landscape complexity}
In this section, we prove Theorem~\ref{prop:many-local-minimum}. 
By definition, $U\subseteq [n]$ is a local minimizer of $H$ if the following conditions hold: 
\begin{align}
\max_{x\in U^c}|E(x, U)|&<\kappa |U|\label{loc1}\\
\min_{x\in U}|E(x, U)|&> \kappa (|U|-1)&\kappa:= \frac{\gamma}{1+\gamma}.\label{loc2}
\end{align} 
Computing the number of local minimizers amounts to counting the number of $U$ satisfying the above conditions. 
More specifically, we count the number of minimizers with fixed sizes, i.e., $|U| = m$ for some $m$.
We have seen in Theorem \ref{thm:global optimum} that local minimizers can only exist in the regime $m= O(\log n)$. 
Hence, we parameterize $m = c\log_2 n$ for some constant $c$, where we use the base $2$ to simplify computation. 
To obtain a lower bound on the number of local minimizers, we only consider $U$'s with no intersection with $\PC$, which take a dominant portion of $m$-subsets of $[n]$ (${n-k\choose m}/{n\choose m}$) and are easier to analyze due to the independence assumptions. 
The following lemma shows that for such $U$'s with no intersection with $\PC$, it suffices to check condition \eqref{loc2} only.  
\begin{Lemma}\label{lm:out}
With probability tending to one, \eqref{loc1} holds for all $U$ with $|U| = c\log_2 n$ for $c>\frac{1}{1-h(\kappa)}$ and $U\cap\PC=\emptyset$.
\end{Lemma}
\begin{proof}
For every $U\subseteq [n]$, let $W(U) = \mathbb I\{\text{$U$ does not satisfy \eqref{loc1}}\}$, and define $X(m)$ as $X(m) = \sum_{U: |U| = m, U\cap [k] = \emptyset}W(U)$. 
By a rapid calculation, one can show that for $|U| = m = c\log_2 n$, with $U \cap \PC = \emptyset$
\begin{align*}
\P(W(U) = 1) &= \prod_{x\in U^c}\P(|E(x, U)|\leq \kappa m) = \left(1- \P(|E(x, U)|> \kappa m)\right)^{n-m}\\
&\stackrel{\kappa>\frac{1}{2}}{\leq} \Big[1- \big(\tfrac{1}{2}\big)^{(1-h(\kappa)+o(1))m}\Big]^{n-m}\leq e^{-\frac{n-c\log_2 n}{n^{c(1-h(\kappa)+o(1))}}}.
\end{align*} 
Hence, when $c>\frac{1}{1-h(\kappa)}$, 
\begin{align*}
\E[X(m)]\leq{n-k\choose c\log_2 n}e^{-\frac{n-c\log_2 n}{n^{c(1-h(\kappa)+o(1))}}} \leq e^{-\frac{n-c\log_2 n}{n^{c(1-h(\kappa)+o(1))}}+c(\log_2 n)^2} = o(1).
\end{align*}
The desired result follows by applying Markov's inequality. 
\end{proof}

We are now ready to prove Theorem \ref{prop:many-local-minimum}.

\begin{proof}[\textbf{\emph{Proof of Theorem \ref{prop:many-local-minimum}}}]
Throughout the proof we ignore the integer rounding effects to simplify discussion. 
Fixing $\frac{1}{1-h(\kappa)}<c<2$, we count the number of size-$m$ subsets with no intersection with $\PC$ and satisfying conditions \eqref{loc1} and \eqref{loc2}. 
By Lemma \ref{lm:out}, \eqref{loc1} holds for all size-$m$ subsets that do not intersect with $\PC$ with probability tending to one, so we only need to check \eqref{loc2}. 
Let $Z(U) = \mathbb I\{\text{$U$ satisfies \eqref{loc2}}\}$, and we are interested in obtaining a lower bound for $Q(m) = \sum_{U: |U|=m, U\cap [k] = \emptyset}Z(U)$, for which we apply a second-moment estimate. 
We first note
\begin{align}
\P\left(Z(U) = 1\right)\geq\P\left(\text{$U$ is a $\kappa m$-regular graph}\right).
\end{align}
To count the number of $\kappa m$-regular graphs among all graph configurations on $m$ vertices, we appeal to a result in \cite{mckay1990asymptotic} that states this number is asymptotically 
\begin{align*}
\sqrt{2}e^{\frac{1}{4}}\left(\frac{1}{2}\right)^{h(\kappa){m\choose 2}}{m-1\choose\kappa m}^m = \sqrt{2}e^{\frac{1}{4}}2^{(h(\kappa)+o(1)){m\choose 2}}.
\end{align*}
Consequently, for all sufficiently large $n$ (hence $m$), 
\begin{align}
\E[Z(U)] = \P\left(Z(U) = 1\right)\geq\left(\frac{1}{2}\right)^{(1-(h(\kappa)+o(1)){m\choose 2}},\label{guoer}
\end{align}
and
\begin{align*}
\E[Q(m)] \geq {n-k\choose m}\left(\frac{1}{2}\right)^{(1-(h(\kappa)+o(1)){m\choose 2}}&\geq\left(\frac{1}{2}\right)^{(1-(h(\kappa)+o(1)){m\choose 2}-m\log_2 (n-k-m)}\\
& = n^{\left(1-\frac{(1-h(\kappa))m}{2\log_2 n} +o(1)\right)m}.
\end{align*}
We next compute the second moment of $Q(m)$. 
\begin{align*}
\E[Q(m)^2] = \sum_{r=0}^m\underbrace{{n-k\choose m}{m\choose r}{n-k-m\choose m-r}}_{:=\xi_r}\underbrace{\E[Z(V_1)Z(V_2)]}_{:=\zeta_r},
\end{align*}
where $(V_1, V_2)$ in the summand is any fixed pair of vertex subsets $V_1,V_2$ with $|V_1\cap V_2| = r, |V_1|=|V_2| = m, V_1\cap [k] = V_2\cap [k] = \emptyset$ (the expectation is independent of the choice of $(V_1, V_2)$ due to exchangeability). 
For $\xi_r$, one can easily check that 
\begin{align}
&\frac{\xi_r}{\xi_0} = (1+o(1))\frac{m^r}{n^r}&r\geq 0.
\end{align}
To analyze $\zeta_r$, we have the following observations. 
For fixed pair $(V_1, V_2)$, we write $\E[Z(V_1) Z(V_2)]= \E[Z(V_1)\E[Z(V_2) \mid E(V_1)]]$ where the conditioning is on the edge configuration on $V_1$. 
By monotonicity of the \ER random graph model, and the fact that~\eqref{loc2} is an increasing event on the edge-set, the conditional expectation $\E[Z(V_2) \mid E(V_1)]$ is maximized when all edges of $V_1$ are present, whence it becomes measurable with respect to the edges on $E\setminus E(V_1)$ and therefore conditionally independent of $Z(V_1)$. As such, 
\begin{align*}
 \mathbb E[Z(V_1)Z(V_2)] \le \mathbb E[Z(V_1)] \mathbb E\Big[Z(V_2) \mid |E(V_1 \cap V_2)| = {r\choose 2}\Big]. 
\end{align*}
By the inequality $\P(A \mid B) \le \P(A) \P(B)^{-1}$, since the event that we have conditioned on has probability $2^{- \binom{r}{2}}$, we get 
\begin{align*}
\zeta_r = \mathbb E[Z(V_1)Z(V_2)]\le \mathbb E[Z(V_1)]\mathbb E[Z(V_2)] 2^{\binom{r}{2}}.
\end{align*}
Consequently, since $\zeta_0 = \E[Z(V_1)]\E[Z(V_2)]$, we have 
\begin{align*}
\frac{\xi_r\zeta_r}{\xi_0\zeta_0}&\leq\frac{m^r2^{{r\choose 2}}}{n^r} = 2^{r(\log_2 m + \frac{r-1}{2}-\log_2 n)}.
\end{align*}
Since $r \le m \le c \log_2 n$ for $c<2$, for $n$ large the quantity in the exponent is negative for $n$ large; thus $r \ge 1$ gives 
\begin{align}
   \frac{\xi_r\zeta_r}{\xi_0\zeta_0} \le  2^{\log_2 m + \frac{c}{2}\log_2 n-\log_2 n}\stackrel{c<2}{\leq} n^{-\frac{1}{2}(1-\frac{c}{2})}.\label{pz}
\end{align}
The proof is finished by appealing to the Paley--Zygmund inequality: for any $\e\in [0, 1]$, we have 
\begin{align*}
\P(Q(m)>\e\E[Q(m)]) & >(1-\e)^2\frac{\E[Q(m)]^2}{\E[Q(m)^2]}= (1-\e)^2\frac{{n-k\choose m}^2\zeta_0}{\sum_{r\in [m]}\xi_r\zeta_r}\\
\geq&\ (1-\e)^2\frac{{n-k\choose m}^2\zeta_0}{(1+(m-1)n^{-\frac{1}{2}(1-\frac{c}{2})})\xi_0\zeta_0}\geq (1-o(1))(1-\e)^2. 
\end{align*}
The proof is completed by taking $\e = \frac{1}{\log_2 n}$.  
\end{proof}

\section{Extension to the robust case}\label{sec:rob}

We now consider a more general setting where a subset of non-PC vertices is allowed to have a higher edge probability.   
Let $\frac{1}{2}\leq q<1$ be fixed parameters, and $k \geq n^{\alpha}, m \geq n^{\lambda}$ for some $0<\alpha, \lambda<1$.  
A contaminated planted clique model $\G(n, \frac{1}{2}, q, k, m)$ with vertices $[n]$ is defined as follows.
One first uniformly samples a $k$-size subset of $[n]$ as the planted clique (denoted by $\PC$) and connects all the edges on it.
Then one arbitrarily picks an $m$-size subset of $[n]\setminus\PC$ (denoted by $\V$) for large-degree vertices and forms the potential edges connected to $\V$ independently with probability $q$.
Finally, one connects the remaining edges independently with probability $\frac{1}{2}$.
When $q = \frac{1}{2}$, $\G(n, \frac{1}{2}, q, k, m)$ reduces to the planted clique model $\G(n, \frac{1}{2}, k)$. 
In general, $\G(n, \frac{1}{2}, q, k, m)$ admits high-degree non-PC vertices to obscure the planted $\PC$ structure, e.g. $m>k$. 
We show that the above algorithms based on the Hamiltonian $H(U)$ of~\eqref{myH} still works when initialized from the full graph to recover $\PC$. 

\begin{Th}\label{lemma:n1+v}
Suppose $\gamma>\frac{1+q}{1-q}$.
For every $\e>0$, there exists $C_0(\e,\gamma), C_1(\e,\gamma)>0$ such that for all $k\geq C_0\sqrt{n}$ and $m^2<\frac{k^3}{C_1\log n}$, with probability at least $1-\e$, the gradient descent $S_t$ initialized from $S_0 = [n]$ achieves 
$$S_t = \PC \qquad \text{ for all \,$t\ge n+2k$}\,.$$
The same holds for the positive temperature chain $S_t^\beta$ for all $n+2k \le t \le n^{k/C_0}$ if $\beta  = \Omega(\log n)$.
\end{Th}

For convenience and without loss of generality, we let $\PC = [k]:=\{1, \ldots, k\}$ and $\V = [k+m]\setminus [k]:=\{k+1, \ldots, k+m\}$ in the subsequent analysis. 
Similar to the analysis of the planted clique model, for any $U\subseteq [n]$, we write $U = U_1\cup U_2\cup U_3$, where $U_1 = U\cap [k]$, $U_2 = U\cap ([k+m]\setminus [k])$, $U_3 = U\cap [k+m]^c$, and $n_1 = |U_1|$, $n_2 = |U_2|$, $n_3 = |U_3|$, and define $H(U_1, U_2, U_3) = H(U)$.
The following lemmas, which are the analogues of Lemmas \ref{lemma:000}, \ref{lemma:000+}, \ref{lemma:monotone}, \ref{lemma:never}, respectively, in the planted clique setting, still hold. 
The proofs are essentially identical, and thus are omitted. 

\begin{Lemma}\label{lemma:00+v}
For any $\gamma>1$, there exists an absolute constant $c(\gamma)>0$ such that, with probability at least $1-o(1)$, for all pairs $(U', U)$ satisfying $U'\subseteq U\subseteq [n]$ with $|U'|\geq c\log n$, 
\begin{align*}
\left(1-\delta\right)\E[\deg(U', U)]\leq \deg(U', U)\leq \left(1+\delta\right)\E[\deg(U', U)],
\end{align*}
where $\deg(U', U) = \sum_{x\in U'} |E(x,U)|$ is the total degree of $U'$ in $U$ and $\delta\geq \sqrt{\frac{96\log n}{|U'|}}$.
Consequently, 
\begin{align}
\min_{x\in U'}|E(x, U)|\leq\left(1+\delta\right)\frac{\E[\deg(U', U)]}{|U'|}.\label{panger+v}
\end{align}
\end{Lemma}

\begin{Lemma}\label{lemma:000+v}
There exists an absolute constant $\rho>0$ such that, for any $\alpha, \lambda\in (0, 1)$ and $k \geq n^\alpha, m\leq n^\lambda$, with probability at least $1-o(1)$, for all $U$ with $|U\cap [k+m]^c|\geq\frac{n}{\log n}$, 
\begin{align*} 
\min_{x\in U\cap [k+m]^c}|E(x, U)|&\leq \frac{\E[\deg(U\cap [k+m]^c, U)]}{|U\cap [k+m]^c|} + \rho\sqrt{n}\\
& = \frac{1}{2}(|U\cap ([k+m]\setminus [k])^c|-1) + q|U\cap ([k+m]\setminus [k])| + \rho\sqrt{n}.
\end{align*}
\end{Lemma}

\begin{Lemma}\label{lemma:monotone+v}
Let $\gamma>\frac{1+q}{1-q}$. 
There exists an absolute constant $c_0(\gamma)>0$ such that with probability $1-o(1)$, for all $U\subseteq [n], |U\cap [k]^c|\geq c_0\log n$,
\begin{align}
\min_{x\in U}H(U\setminus\{x\})\leq\min\left\{H(U), \min_{z\in U^c}H(U\cup\{z\})\right\} - 1. \label{happy+v}
\end{align}
Consequently, $|S_t|$ will only decrease until $|S_t\cap [k]^c|\leq c_0\log n$. 
\end{Lemma}

\begin{Lemma}\label{lemma:never+v}
Let $\gamma>\frac{1+q}{1-q}$ and $\xi<1-q(1+\frac{1}{\gamma})$. With probability $1-o(1)$ the following holds. For any $S$ having $|S\cap [k]|\geq \max\{\gamma|S\cap [k]^c|+2, (1-\xi)k\}$, if $S_t'$ is the gradient descent initialized from $S_0'= S$ then $d_\H(S_t', \PC)$ is strictly decreasing in $t$ while $S_t'\ne \PC$. Consequently, $S_t'$ will converge to $\PC$ in at most $2k$ further steps. 

Moreover, for any $W\subseteq [n], W\neq\PC$ that satisfies $|W\cap [k]|\geq \max\{\gamma|W\cap [k]^c|+2, (1-\xi)k\}$, 
\begin{align}
\min_{U\in\mathcal U}H(U)\leq \min_{U\sim W, U\notin\mathcal U}H(U) - 1, \label{egs+v}
\end{align} 
where $\mathcal U = \left\{U\sim W: d_\H(U, \PC)<d_\H(W, \PC)\right\}$ represents the set of neighboring states of $W$ that are one Hamming distance closer to the $\PC$ than $W$ is.  
\end{Lemma}

Let $Y_t$ be the same peeling process introduced in Section~\ref{sec:full-graph} and $\tau_0=\tau_0(c_0)$ be the first time $|Y_t\cap [k]^c|\leq c_0\log n$; see the paragraph above \eqref{xstar} for the details of their definitions. 
To establish a similar version of Lemma \ref{lemma:A}, we need the following stochastic dominance result concerning the dynamics of $Y_t$ while running on the contaminated planted clique model.  
For convenience, we introduce the following notation to keep track of the statistics of the dynamics:
\begin{align*}
U_{t, 1}& = Y_{t}\cap [k]&\quad&n_{t, 1} = |U_{t, 1}|\,,\\
U_{t, 2} &= Y_{t}\cap ([k+m]\setminus [k])&\quad& n_{t, 2} = |U_{t, 2}|\,,\\
U_{t, 3} &= Y_{t}\cap [k+m]^c&\quad& n_{t, 3} = |U_{t, 3}|\,,\\
\bar{U}_{t, 1}& = Y^c_{t}\cap [k]&\quad&\bar{n}_{t, 1} = |\bar{U}_{t, 1}|\,,\\
\bar{U}_{t, 2} &= Y^c_{t}\cap ([k+m]\setminus [k])&\quad& \bar{n}_{t, 2} = |\bar{U}_{t, 2}|\,,\\
\bar{U}_{t, 3} &= Y^c_{t}\cap [k+m]^c&\quad& \bar{n}_{t, 3} = |\bar{U}_{t, 3}|\,.
\end{align*}
For $x\in [n]$, let $T_x$ denote the time that $x$ is removed from $Y_t$ before $\tau_0$, i.e., $T_x = \min\{t: x\notin Y_t\}\wedge \tau_0$. 
The below lemma is saying that for any $x\in [k]$ and before its removal from $Y_t$, both $\{\deg(x, \bar{U}_{t, 2})\}_{t\leq T_x}$ and $\{\deg(x, \bar{U}_{t, 3})\}_{t\leq T_x}$ are upper bounded by the respective correct Binomial processes under appropriate coupling.

\begin{Lemma}\label{lemma:sd}
Let $Y_t$ be the peeling process and $\mathcal F_t$ be the associated filtration. 
For every $x\in [k]$, there exist processes $\{\bar{b}_{t,2}(x)\}_{t< T_x}$, $\{\bar{b}_{t,3}(x)\}_{t< T_x}$ defined on the same probability space of $\{\deg(x, \bar{U}_{t, 2})\}_{t< T_x}$, $\{\deg(x, \bar{U}_{t, 3})\}_{t< T_x}$ and upper bound them, respectively, where 
\begin{align}
&\bar{b}_{t,2}(x) = \sum_{i=0}^{\bar{n}_{t,2}}\e_i&\bar{b}_{t,3}(x) = \sum_{i=0}^{\bar{n}_{t,3}}\e_i',\label{veryhappy}
\end{align}
and $\{\e_i\}_{i\geq 1}, \{\e_i'\}_{i\geq 1}$ are independent Bernoulli random variables with parameters $q$ and $1/2$, respectively, with initial values $\e_0 = \e_0' = 0$. 
\end{Lemma}
\begin{proof}
The proof follows from a similar edge-exposure argument of \cite[Lemma 4]{feige2010finding} and we only provide a sketch here for completeness.
For convenience, we focus on the stopped process $\{\deg(x, \bar{U}_{t, 2})\}_{t< T_x}$, the argument for $\{\deg(x,\bar{U}_{t,3})_{t<T_x}$ being identical. 

Fixing $x\in [k]$, we first draw an edge configuration of the contaminated model except for the edges between $x$ and $[k+m]\setminus [k]$. 
We call such a configuration partially observed.  
We run a \emph{tentative} peeling process $Y_t$ on the partially observed configuration as follows. 
Suppose that at time $t<T_x$, $z$ is chosen as the candidate for removal because it has the lowest degree in $Y_t$ among the revealed portion of the graph ($z$ could be equal to $x$).
If all edges adjacent to $z$ have been revealed, then remove $z$ as usual. 
Otherwise, randomly choose one of the unrevealed edges adjacent to $z$ and reveal it (this corresponds to drawing an independent Bernoulli random variable Ber($q$) to decide whether to form an edge or not), and then re-find the smallest-degree vertex in $Y_t$ with this additional information. 
Note that all unrevealed edges at the beginning have exactly one vertex $x$ and the other vertex $z\in[k+m]\setminus [k]$.
For each $z\in[k+m]\setminus [k]$, we let $\widetilde{\e}_z$ denote the edge outcome between $x$ and $z$ after revealing, i.e., $\widetilde{\e}_z = 1$ if there is an edge between $x$ and $z$ and $\widetilde{\e}_z = 0$ if there is not. 
The tentative peeling process agrees with the original peeling process before $T_x$. 

At each time $t< T_x$, the edge information between $x$ and every vertex in $\bar{U}_{t,2}$ must have been revealed. Denoting all revealed vertices in $[k+m]\setminus [k]$ at time $t$ as $\mathcal R_t$, it is easy to see that
$\left\{z\in\mathcal R_t: \widetilde{\e}_z=0\right\}\subseteq\bar{U}_{t,2}$ so $\{z\in \mathcal R_t: z\notin \bar U_{t,2}\}\subseteq \{z\in \mathcal R_t: \widetilde{\varepsilon}_z =1\}$.
Consequently, 
\begin{align*}
\deg(x, \bar{U}_{t, 2}) = \sum_{z\in \bar{U}_{t, 2}}\widetilde{\e}_{z} = \sum_{z\in\mathcal R_t}\widetilde{\e}_{z}-(|\mathcal R_t|-\bar{n}_{t,2})\leq\sum_{i=0}^{\bar{n}_{t,2}}\widetilde{\e}_{z^{(i)}},
\end{align*}
where $z^{(i)}$ denotes the $i$-th vertex in $\mathcal R_t$ whose edge to $x$ was revealed, with the convention that $\tilde{\e}_{z^0}=0$ identically.
By definition, $\{\widetilde{\e}_{z^{(i)}}\}_{i\in [\bar{n}_{t,2}]}$ are i.i.d. Bernoulli random variables with parameter $q$. 
Letting $\e_i = \widetilde{\e}_{z^{(i)}}$ yields the desired result. 
\end{proof}

Using Lemma \ref{lemma:sd} and applying the Optional Stopping Theorem and Doob's inequality, we can prove the following result which can be viewed as a more general version of Lemma \ref{lemma:A}. 

\begin{Lemma}\label{lemma:A+v}
Let $Y_t$ denote the peeling process defined above. 
For any $\e, \eta>0$, there exists an absolute constant $c_1(\e, \eta)>0$ such that if we define the random subset of $[k]$, 
\begin{align}
\mathcal A(c_1) = \left\{x\in [k]: \deg(x, Y_t)\geq (n_{t,1}-1) + qn_{t,2} + \frac{1}{2}n_{t,3} - c_1\sqrt{n}\ \text{for all $t<T_x$}\right\}.\label{mya+v}
\end{align} 
 then with probability at least $1-\e$, $\left|\mathcal A\right|\geq (1-\eta)k$.
\end{Lemma}
\begin{proof}
Fix $x\in [k]$ and let $\bar{b}_{t,2}(x)$ be the process defined in \eqref{veryhappy} in Lemma \ref{lemma:sd} that bounds $\deg(x, \bar{U}_{t, 2})$. 
Observe that the process $\bar{b}_{t,2}(x)- q \bar n_{t,2}$ is a martingale with bounded increments. 
Applying the Optional Stopping Theorem (OST), Doob's maximal inequality to the centered dominating process, and Lemma \ref{lemma:sd}, we have 
\begin{align*}
&\P\left(\max_{t<T_x}(\deg(x, \bar{U}_{t, 2})-q \bar{n}_{t, 2})>c_{11}\sqrt{n}\right)\stackrel{\text{Lemma \ref{lemma:sd}}}{\leq}\P\left(\max_{t<T_x}(\bar{b}_{t,2}(x)-q \bar{n}_{t, 2})>c_{11}\sqrt{n}\right)\\
& \qquad \stackrel{\text{OST+Doob}}{\leq}\frac{\E[(\bar{b}_{T_x-1,2}(x)-q \bar{n}_{T_x-1, 2})^2]}{c_{11}^2 n}\stackrel{\bar{n}_{T_x-1, 2}\leq n}{\leq} \frac{q(1-q)}{c_{11}^2}\leq \frac{\e\eta}{3},
\end{align*}
where the last inequality holds if $c_{11}^2\geq\frac{q(1-q)}{\e}$. 
For the same $c_{11}$, by the same reasoning, we have 
\begin{align*}
\P\left(\max_{t<T_x}(\deg(x, \bar{U}_{t, 3})-\bar{n}_{t, 3}/2)>c_{11}\sqrt{n}\right)\leq \frac{\e\eta}{3}.
\end{align*}
Putting the two cases together and noting $\deg(x, \bar{U}_{t, 1})\leq \bar{n}_{t,1}$ yields that, with probability at most $1-\frac{2\e\eta}{3}$, 
\begin{align}
\max_{t<T_x}\left(\deg(x, \bar{U}_{t, 1})+\deg(x, \bar{U}_{t, 2})+\deg(x, \bar{U}_{t, 3})-\bar{n}_{t,1} - q\bar{n}_{t,2}-\bar{n}_{t, 3}/2\right)\leq 2c_{11}\sqrt{n}.\label{213}
\end{align}
Meanwhile, since $\deg(x)$ is a sum of independent Bernoulli random variables, by Chebyshev's inequality, there exists $c_{21}(\e, \eta)>0$ such that with probability at least $1-\frac{\e\eta}{3}$,  
\begin{align}
\deg(x)\geq k-1 + qm + \frac{n-k-m}{2}-c_{21}\sqrt{n}.\label{214}
\end{align} 
It is easy to check that \eqref{213} and \eqref{214} together implies the property defined by $\A$ in \eqref{mya+v} with $c_1 = 2c_{11}+c_{21}$, and such a property holds with probability at least $1-\e\eta$ for each $x\in [k]$.  
The desired follows by an application of Markov's inequality. 
\end{proof}

We are ready to prove the success result for the robust version of the planted clique problem.

\begin{proof}[Proof of Theorem \ref{lemma:n1+v}]
Given Lemmas~\ref{lemma:sd}--\ref{lemma:A+v}, the proof is similar to the proof of Theorem \ref{mainthm:success}, and we only describe the differences that arise in the gradient descent case.
By similar reasoning to the proof of Theorem~\ref{mainthm:success}, it suffices to verify that for any $\eta<1-q(1+\frac{1}{\gamma})<\frac{1}{2}$, there exists $c_1(\eta)>0$ such that with probability at least $1-\e$, 
\begin{align*}
|\A|\geq (1-\eta)k \quad \text{and} \quad \A\subseteq Y_t \quad \text{for all $t\leq\tau_0$},
\end{align*}
where $c_1$ appears in the definition of $\A$ in \eqref{mya+v}. 
To this end, we apply an inductive argument.  
Let $k \geq C_0\sqrt{n}$ for some absolute constant $C_0$ to be determined during the proof.
Taking a union bound for the statements in Lemmas \ref{lemma:00+v}, \ref{lemma:000+v}, \ref{lemma:A+v}, we have the following events hold simultaneously with probability at least $1-\e$: there exist $c, c_1, \rho>0$ such that 
\begin{align}
\min_{x\in U'}|E(x, U)|&\leq\left(1+\delta\right)\frac{\E[\deg(U', U)]}{|U'|}\qquad \forall U'\subseteq U\subseteq [n], |U'|\geq c\log n\label{t1+v}\\
\min_{x\in U\cap [k+m]^c}|E(x, U)|&\leq  \frac{1}{2}(|U\cap ([k+m]\setminus [k])^c|-1) + q|U\cap ([k+m]\setminus [k])| + \rho\sqrt{n}\label{t2+v}\\
&\ \ \ \ \ \ \ \ \ \ \ \ \ \ \ \ \ \ \ \ \ \ \ \ \ \ \ \ \ \ \ \ \ \ \ \ \ \forall U\subseteq [n], |U\cap [k+m]^c|\geq\frac{n}{\log n}\nonumber\\
\left|\mathcal A\right|=\left|\mathcal A(c_1)\right|&\geq (1-\eta)k\label{t3+v},
\end{align}
where $\delta = \delta(U') = \sqrt{\frac{96\log n}{|U'|}}$. 
It is easy to see that $\A\subseteq Y_0 = [n]$ at $t=0$.
To apply induction, we assume $\A\subseteq Y_{t-1}$ and verify $\A\subseteq Y_{t}$. 
We now consider three separate cases. 

If $n_{t-1, 3}\geq \frac{n}{\log n}$, letting $U = Y_{t-1}$ in \eqref{t2+v}, we have
\begin{align*}
\min_{y\in U_{t-1, 3}}& \deg(y, Y_{t-1}) \stackrel{\eqref{t2+v}}{\leq} \frac{1}{2}(|Y_{t-1}|-1) + \rho\sqrt{n}\\
\leq&\ n_{t-1, 1}-1 + qn_{t-1, 2} + \frac{1}{2}n_{t-1, 3} - c_1\sqrt{n} - \left[\frac{1}{2}(n_{t-1, 1}-1) - \rho\sqrt{n} - c_1\sqrt{n}\right]\\
\stackrel{\eqref{mya+v}}{\leq}&\ \min_{x\in \A\subset Y_{t-1}}\deg(x, Y_{t-1}) -  \left[\frac{1}{2}(n_{t-1, 1}-1) - \rho\sqrt{n} - c_1\sqrt{n}\right]\\
\stackrel{\eqref{t3+v}}{\leq}&\ \min_{x\in \A\subset Y_{t-1}}\deg(x, Y_{t-1}) -  \left[\frac{C_0(1-\eta)}{2}\sqrt{n} - \rho\sqrt{n} - c_1\sqrt{n}\right]<\min_{x\in \A\subset Y_{t-1}}\deg(x, Y_{t-1}),
\end{align*}
where the last inequality holds for any $C_0$ satisfying $C_0>\frac{2(\rho+ c_1)}{1-\eta}$. 

If $n_{t-1, 2}\leq n_{t-1, 3}$ or $\frac{K m^2\log n}{k^2}\leq n_{t-1, 3}<\frac{n}{\log n}$ for some large absolute constant $K$ (e.g. $K = 10^5$ works), then letting $(U, U') = (Y_{t-1}, U_{t-1, 3})$ in \eqref{t1+v}, we have
\begin{align*}
\min_{y\in U_{t-1, 3}}&\deg(y, Y_{t-1})\stackrel{\eqref{t1+v}}{\leq}(1+\delta)\left[\frac{1}{2}(n_{t-1, 1} + n_{t-1, 3}-1) + qn_{t-1, 2}\right]\\
\leq&\  n_{t-1, 1}-1 + qn_{t-1, 2} + \frac{1}{2}n_{t-1, 3} - c_1\sqrt{n} - \left[\frac{1-\delta}{2}(n_{t-1, 1}-1) - \delta (qn_{t-1, 2}+n_{t-1, 3}) - c_1\sqrt{n} \right]\\
\stackrel{\eqref{mya+v}}{\leq}&\  \min_{x\in \A\subset Y_{t-1}}\deg(x, Y_{t-1}) - \left[\frac{1-\delta}{2}(n_{t-1, 1}-1) - \delta (qn_{t-1, 2}+n_{t-1, 3}) - c_1\sqrt{n} \right]\\
\leq\ & \min_{x\in \A\subset Y_{t-1}}\deg(x, Y_{t-1}) - \left(\frac{k}{8} - \frac{k}{16} - 2\sqrt{96 n} - c_1\sqrt{n} \right)\\
<\ &\min_{x\in \A\subset Y_{t-1}}\deg(x, Y_{t-1}),
\end{align*}
where the second to last inequality follows by noting that $n_{t-1, 1}\geq |\mathcal A|\geq (1-\eta)k>\frac{k}{2}$ and under the conditions $n_{t-1, 2}\leq n_{t-1, 3}$ or $\frac{K m^2\log n}{k^2}\leq n_{t-1, 3}<\frac{n}{\log n}$, 
\begin{align*}
\delta(n_{t-1, 2}+n_{t-1, 3})&\leq\max\left\{2\sqrt{96n_{t-1,3}\log n}, \sqrt{96n_{t-1,3}\log n}+ \sqrt{\frac{96n^2_{t-1,2}\log n}{n_{t-1, 3}}}\right\}\\
&\leq \max\left\{2\sqrt{96 n}, \sqrt{96 n}+ \frac{k}{16}\right\}\leq 2\sqrt{96 n}+ \frac{k}{16},
\end{align*}
and the last inequality holds if $\frac{C_0}{16} - 2\sqrt{96} - c_1>0$.

Finally, if $n_{t-1, 2}> n_{t-1, 3}$ and $n_{t-1, 3}< \frac{K m^2\log n}{k^2}$ where $K$ is the same as in the previous case, then letting $(U, U') = (Y_{t-1}, U_{t-1, 2})$ in \eqref{t1+v}, we have
\begin{align*}
\min_{y\in U_{t-1, 2}}&\deg(y, Y_{t-1})\leq (1+\delta)q(n_{t-1, 1} + n_{t-1, 2} + n_{t-1, 3}-1)\\
\leq&\  n_{t-1, 1}-1 + qn_{t-1, 2} + \frac{1}{2}n_{t-1, 3} - c_1\sqrt{n} - \left[(1-(1+\delta)q)(n_{t-1, 1}-1) - n_{t-1, 3} - 2\delta qn_{t-1, 2} - c_1\sqrt{n} \right]\\
\stackrel{\eqref{mya+v}}{\leq}&\ \min_{x\in \A\subset Y_{t-1}}\deg(x, Y_{t-1}) - \left[(1-(1+\delta)q)(n_{t-1, 1}-1) - n_{t-1, 3} - 2\delta qn_{t-1, 2} - c_1\sqrt{n} \right]\\
\leq\ & \min_{x\in \A\subset Y_{t-1}}\deg(x, Y_{t-1}) - \left[\frac{(1-q)k}{4} -\frac{(1-q)k}{8} - 2\sqrt{96 n_{t-1, 2}} - c_1\sqrt{n} \right]\\
<\ &\min_{x\in \A\subset Y_{t-1}}\deg(x, Y_{t-1}),
\end{align*}
where the second to last inequality follows by noting $n_{t-1, 1}\geq |\mathcal A|\geq (1-\eta)k>\frac{k}{2}$, and $k^3\geq C_1m^2\log n$ with $C_1 = \frac{8K}{C_0 (1-q)}$ implies
\begin{align*}
n_{t-1, 3}< \frac{K m^2\log n}{k^2}\leq\frac{(1-q)k}{8}, 
\end{align*}
and the last inequality holds if $\frac{C_0(1-q)}{8} - 2\sqrt{96} - c_1>0$. 
\end{proof}

\subsection*{Acknowledgements}
The authors thank anonymous referees for their helpful feedback. The research of R.G. was supported in part by NSF DMS-2246780. A.J. acknowledges the support of the Natural Sciences and Engineering Research Council of Canada (NSERC) and the Canada Research Chairs programme. Cette recherche a \'et\'e enterprise gr\^ace, en partie, au 
soutien financier du Conseil de Recherches en Sciences Naturelles et en G\'enie du Canada (CRSNG),  [RGPIN-2020-04597, DGECR-2020-00199], et du Programme des chaires de recherche du Canada.
The authors thank Kailong Zhang for helping make the phase diagram plot  Figure~\ref{fig:1}.

\bibliographystyle{amsalpha}
\bibliography{ref}
\end{document}